\documentclass[11pt, reqno, oneside]{article}
\usepackage{geometry} % to change the page dimensions
\usepackage[affil-it]{authblk}

\usepackage{amsmath, amsthm, amsfonts, amssymb, amsbsy, bigstrut, graphicx, enumerate,  upref, longtable, comment, booktabs, array, caption, subcaption}

\usepackage[breaklinks]{hyperref}

\usepackage[numbers, sort&compress]{natbib}

\newcommand{\me}{\mathcal{E}}

\newtheorem{thm}{Theorem}%[section]
\newtheorem{lmm}{Lemma}
\newtheorem{cor}{Corollary}

\theoremstyle{definition}

\newcommand{\cov}{\mathrm{Cov}}

\newcommand{\ee}{\mathbb{E}}

\newcommand{\mf}{\mathcal{F}}

\newcommand{\cp}{\mathcal{P}}

\newcommand{\pp}{\mathbb{P}}

\newcommand{\rr}{\mathbb{R}}
\newcommand{\smallavg}[1]{\langle #1 \rangle}
\newcommand{\avg}[1]{\langle #1 \rangle}

\newcommand{\var}{\mathrm{Var}}

\newcommand{\zz}{\mathbb{Z}}

\newcommand{\fpar}[2]{\frac{\partial #1}{\partial #2}}

\newcommand{\mpar}[3]{\frac{\partial^2 #1}{\partial #2 \partial #3}}

\usepackage[utf8]{inputenc}
\usepackage{amsmath}
\usepackage{amsfonts}
\usepackage{bbm}
\usepackage{mathtools}
\usepackage{tikz}
\usetikzlibrary{decorations.pathreplacing}
\usepackage{amsthm}
\usepackage[english]{babel}
\usepackage{fancyhdr}
\usepackage{amssymb}
\usepackage{enumitem}
\usepackage{qtree}
\usepackage{mathrsfs}

\newcommand{\N}{\mathbb{N}}
\newcommand{\Z}{\mathbb{Z}}

\newcommand{\R}{\mathbb{R}}

\newcommand{\E}{\mathbb{E}}

\renewcommand{\P}{\mathbb{P}}

\renewcommand{\tilde}{\widetilde}
\renewcommand{\hat}{\widehat}

\begin{document}
%\title{Local KPZ growth in a class of random surfaces}
%[Recent developments in measures of association]
%\title{Replica symmetry in short range spin glasses}
\title{Spin glass phase at zero temperature in the Edwards--Anderson model}
\author{Sourav Chatterjee\thanks{Department of Statistics, Stanford University, 390 Jane Stanford Way, Stanford, CA 94305, USA. Email: \href{mailto:souravc@stanford.edu}{\tt souravc@stanford.edu}. %The author was partially supported by NSF grants DMS-2113242 and DMS-2153654. %The author thanks xxx for helpful comments.
}}
\affil{Stanford University}
%\address{Departments of mathematics and statistics, Stanford University}
%\email{souravc@stanford.edu}
%\dedicatory
%Research partially supported by NSF grants DMS-1855484 and DMS-2113242}
%\keywords{}
%\subjclass[2020]{}

\maketitle

%\begin{center}
%{\it \small In honor of friend and teacher Prof.~Rajeeva L.~Karandikar on the occasion of his 65$^{th}$ birthday.}
%\end{center}

\begin{abstract}
Mean field spin glass models have undergone substantial mathematical development, but finite dimensional short range spin glasses remain much less understood. This paper proves several rigorous zero temperature signatures of glassy behavior for the Edwards--Anderson model with Gaussian couplings, in finite boxes in arbitrary dimension. First, the ground state is sensitive to small perturbations of the disorder: after a perturbation of size $p$, the new ground state is nearly orthogonal to the original one in site overlap once $p$ is sufficiently larger than the inverse system size. Second, the droplets generated by such perturbations have large interfaces; in the macroscopic-droplet regime, their boundaries satisfy lower bounds consistent with a fractal dimension strictly greater than $d-1$. Third, there exist macroscopic spin excitations whose energy cost is negligible compared with the size of their interface, in sharp contrast with ferromagnetic behavior. Fourth, the expected size of the critical droplet associated with a typical bond grows at least as a power of the volume. Finally, a natural boundary condition sensitivity for nearest-neighbor spin products cannot decay faster than order the inverse distance to the boundary, contrasting with recent exponential decay results for the two-dimensional random field Ising model. Taken together, these results provide rigorous evidence --- and, in the senses made precise below, proof --- of zero temperature glassy behavior in a short range spin glass model.
\end{abstract}

%\tableofcontents

%\section{A new coefficient of correlation}
\section{Introduction}
Spin glasses are disordered magnetic systems whose low temperature behavior is characterized by frustration, many competing nearly optimal configurations, and sensitivity to small changes in the disorder. Experimental examples include dilute alloys such as AuFe and CuMn. Their mathematical analysis is challenging, especially for finite dimensional models with short range interactions.

Spin glass models fall broadly into two classes: mean field models, such as the Sherrington--Kirkpatrick (SK) model, and finite dimensional lattice models, such as the Edwards--Anderson (EA) model. The EA model was introduced in \cite{edwardsanderson75} as a finite dimensional model of a spin glass with short range interactions. In contrast to the SK model \cite{sherringtonkirkpatrick75}, whose mathematical analysis has seen major progress \cite{talagrand10, talagrand11, panchenko13}, the EA model remains much less understood. A central question is whether the EA model exhibits, at low enough temperature, the phenomena usually associated with a spin-glass phase. This question is not merely one of mathematical rigor: even in the physics literature, the low temperature structure of the EA model remains the subject of competing predictions. For more on this longstanding debate, see \cite{braymoore85, fisherhuse88, krzakalamartin00, mcmillan84, mezardetal87, newmanstein03, palassiniyoung00, contuccietal06, marinariparisi01, contuccietal07, contuccigiardina13, read14} and the references therein.

This paper proves several rigorous zero temperature results for the EA model with Gaussian disorder. These include disorder chaos for the site overlap, lower bounds on the size of droplet interfaces, the existence of low energy macroscopic excitations, polynomial lower bounds on critical droplet sizes, and slow decay of boundary condition dependence for nearest neighbor spin products. Together, these results give rigorous evidence for zero temperature glassy behavior in a short range spin glass model, and they sharply distinguish the EA model from ferromagnetic systems.

We emphasize that the results in this paper are finite volume, zero temperature statements. They do not settle the infinite volume question of whether the Edwards--Anderson model has incongruent ground states, nor do they distinguish conclusively between all competing pictures of the low temperature phase. Rather, they establish several rigorous signatures that have long been associated with spin glass behavior in the physics literature: disorder chaos, large and rough droplet interfaces, low energy macroscopic excitations, growth of critical droplets, and slow decay of boundary condition sensitivity. In this precise sense, the results give a rigorous proof of zero temperature glassy behavior for a short range spin glass model.

We begin by defining the model and then state the results, grouped by theme. The final subsection of the introduction discusses related literature and open problems. Proofs are given in Section~\ref{proofsec}.

\subsection{The model}
Let $G$ be a finite, simple, connected graph with vertex set $V$ and edge set $E$. Let $J = (J_e)_{e\in E}$ be a collection of i.i.d.~random variables with a given law $\mu$. The Edwards--Anderson Hamiltonian on this graph in the environment (or disorder, or bond strengths, or edge weights) $J$ is the random function $H_J :\{-1,1\}^V \to \R$ defined as
\[
H_J(\sigma) := -\sum_{\{i,j\}\in E} J_{ij} \sigma_i \sigma_j.
\]
A ground state for this model is a state $\sigma$ (depending on $J$) that minimizes the above Hamiltonian. If $\mu$ has no atoms, then it is not hard to show that with probability one, there are exactly two ground states $\sigma$ and $-\sigma$. 

What we have described above is the ground state under the free boundary condition. Sometimes we impose a boundary condition, in  the following way. Let $B$ be a nonempty subset of $V$ and $\gamma$ be a fixed element of $\{-1,1\}^B$. Then the ground state under boundary condition $\gamma$ on the boundary $B$ is the minimizer of $H_J(\sigma)$ under the constraint that $\sigma_i = \gamma_i$ for each $i\in B$. Again, it is not hard to show that under a boundary condition, there is a unique ground state with probability one if $\mu$ has no atoms, provided that $V\setminus B$ is a connected subset of $V$. {\it We will henceforth assume that $V\setminus B$ is connected.}

For concreteness, one may keep in mind the nearest-neighbor box
\[
G_L=\{0,1,\ldots,L\}^d\subseteq \Z^d .
\]
In the absence of a boundary condition, this is the finite-volume EA model with free boundary condition. The usual boundary $B_L$ is the set of vertices with at least one coordinate equal to $0$ or $L$. Another standard choice is periodic boundary condition, obtained by identifying vertices on opposite faces of the box.

The EA model at inverse temperature $\beta$ assigns a probability measure with mass proportional to $e^{-\beta H(\sigma)}$ at each $\sigma$. The $\beta =\infty$ (zero temperature) model is just the probability measure that puts all its mass on the ground state (or the uniform distribution on the pair of ground states in the free boundary case).  In this paper, we will only consider the zero temperature model. Also, throughout, {\it we will take the disorder distribution $\mu$ to be the standard Gaussian distribution,} although various parts of the proofs should work for quite general distributions. 

One source of difficulty with understanding properties of the ground state is computational. At zero temperature, finding a ground state is equivalent to solving a weighted cut optimization problem, or equivalently, after changing the signs of the weights, a weighted maximum cut problem. The maximum cut problem is NP-hard for general graphs \cite{gareyjohnson79}, although polynomial time algorithms are available for planar graphs \cite{shihetal90}. Thus, even the exact finite volume optimization problem underlying the model is computationally hard in general.

%\subsection{Results}
\subsection{Chaotic nature of the ground state} 
Our first main result is that the ground state of the EA model with standard Gaussian disorder is sensitive to small changes in the disorder $J$, a phenomenon that is sometimes called ``disorder chaos''. We consider two kinds of perturbations, both determined by a parameter $p\in (0,1)$. In the first kind of perturbation, we replace each $J_e$ by $(1-p)J_e + \sqrt{2p-p^2} J_e'$, where $J' = (J_e')_{e\in E}$ is another set of i.i.d.~standard Gaussian random variables, independent of $J$. The coefficients in front of $J_e$ and $J_e'$ are chosen to ensure that the linear combination is again a standard Gaussian random variable. In the second kind of perturbation, each $J_e$ is replaced by $J_e'$ with probability $p$, independently of each other. 

Let $V^\circ := V\setminus B$ denote the set of ``interior vertices'' of $V$. Note that $V^\circ = V$ when $B=\emptyset$ (the case of free boundary). We have already assumed earlier that $V^\circ$ is connected. To avoid trivialities, {\it we will assume that $V^\circ$ is nonempty and $|E|\ge 2$.}  Let $\sigma$ be the ground state in the original environment and $\sigma'$ be the ground state in the perturbed environment. The ``site overlap'' between the two configurations is defined as
\[
R(p) := \frac{1}{|V^\circ|} \sum_{i\in V^\circ} \sigma_i \sigma_i'.
\]
If $B=\emptyset$ (i.e., for the free boundary condition), $R(p)$ is not well-defined since there are two ground states in both environments. But $R(p)^2$ is still well-defined, and  that is sufficient for our purposes. Note that $R(p)$ is close to zero if and only if $\sigma$ and $\sigma'$ are nearly orthogonal to each other --- or in other words, $\sigma$ and $\sigma'$ disagree on approximately half the vertices. The following theorem shows that under certain conditions, $R(p)\approx 0$ with high probability for a tiny value of $p$, which is what's commonly known as disorder chaos for the site overlap. We first state the result  for a general graph $G$, and then specialize to the case of a cube in $\Z^d$ in the corollary that follows. The proof is in Section \ref{mainthmpf}.
\begin{thm}\label{mainthm}
Let all notations be as above. Let $d$ denote the graph distance on $G$. Suppose that there are positive constants $\alpha$, $\beta$, $\gamma$ and $\delta$ such that for any $i\in V^\circ$ and $r\ge 1$, the number of $j$ such that $d(i,j)\le r$ is at most $\alpha r^\beta$, and the number of $j$ such that $\min\{d(j,k):k\in B\}\le r$ is at most $\gamma |B| r^\delta$. Then for both kinds of perturbations, we have that for any $p\in (0,1)$, 
\begin{align*}
\ee(R(p)^2) &\le  \frac{C (|V^\circ| p^{-\beta} + |B|^2 p^{-2\delta})}{|V^\circ|^2}, %\frac{1}{|V^\circ|} +
\end{align*}
where $C$ is a constant depending only on $\alpha$, $\beta$, $\gamma$ and $\delta$.
\end{thm}
Let us now check what this yields for $V = \{0,1,\ldots,L\}^d$ with the usual boundary, for some dimension $d\ge 1$ (not to be confused with the graph distance $d$). In this case, $|V^\circ|$ is of order $L^d$, $|B|$ is of order $L^{d-1}$, $\beta = d$, and $\delta = 1$. Thus, we get the following corollary.
\begin{cor}\label{chaoscor}
If $V = \{0,1,\ldots,L\}^d$ with the usual boundary and with any given disorder-independent boundary condition, then for both kinds of perturbations, we have that for all $p\in (0, 1)$, 
\[
\ee(R(p)^2) \le 
\begin{cases}
C(d)L^{-1} p^{-1} &\text{ if } d=1,\\
C(d)L^{-2}p^{-2} &\text{ if } d\ge 2,
\end{cases}
\]
where $C(d)$ depends only on $d$. For  free or periodic boundary, the bound is
\[
\ee(R(p)^2) \le \frac{C(d)}{L^d p^d} \ \ \text{ for all } d\ge 1. 
\]
\end{cor} 
Thus, $R(p)$ is small with high probability whenever $p\gg L^{-1}$. Equivalently, under such perturbations the original and perturbed ground states disagree on approximately half of the sites. This proves disorder chaos for the site overlap at zero temperature, as predicted in \cite{braymoore87}. The threshold $L^{-1}$ may not be optimal; simulations in \cite{braymoore87} suggest that smaller perturbations may already produce chaos.

The proof of Theorem~\ref{mainthm} also gives a more local form of disorder chaos for relative spin orientations. This is close to the prediction in \cite{braymoore87} that, in a glassy phase, relative orientations of widely separated spins should be sensitive to small changes in the bond strengths.

\begin{thm}\label{chaosthm}
In the setting of Theorem \ref{mainthm}, take any $p\in (0,1)$, and let $\sigma$ and $\sigma'$ be the ground states of the unperturbed system and the system with perturbation parameter $p$ (for either kind of perturbation), respectively. Then for any $i,j\in V^\circ$,
\[
%\biggl|\pp(\sigma_i \sigma_j \ne \sigma_i'\sigma_j')-\frac{1}{2}\biggr|\le \frac{1}{2} 
|\ee(\sigma_i\sigma_j \sigma_i'\sigma_j')| \le (1-p)^{\min\{d(i,j), d(i,B)+d(j,B)\}},
\]
where $d(i,B) := \min\{d(i,k): k\in B\}$ (defined to be infinity if $B=\emptyset$). 
\end{thm}
This theorem, proved in Section \ref{chaosthmpf}, shows that if $i$ and $j$ are two vertices such that $d(i,j)$, $d(i,B)$ and $d(j,B)$ are all much greater than $p^{-1}$, then the relative orientations of the spins at $i$ and $j$ in the original and the perturbed environments are approximately independent of each other (since marginally, both $\sigma_i \sigma_j$ and $\sigma_i'\sigma_j'$  are uniformly distributed on $\{-1,1\}$).

Notice the contrast between the EA model and any ferromagnetic model --- even one with random bonds --- in Theorems \ref{mainthm} and \ref{chaosthm}. In a ferromagnetic model, a small perturbation of the environment does not change the ground state at all, whereas in the EA model, a small perturbation causes such a large change that the original and perturbed ground states are almost orthogonal to each other. 

\subsection{Fractal dimension of the droplet boundary}
Consider the EA model on $V= \{0,1,\ldots, L\}^d$. Suppose we apply a perturbation of the first kind, with perturbation parameter $p$. Let $A$ be the region of overturned spins, which is usually called a ``droplet''. Suppose that the perturbation is sufficiently large to ensure that $A$ has macroscopic  size (i.e., of order $L^d$).  What is the size of the edge boundary $\partial A$ of $A$ (i.e., the set of all edges from $A$ to $V\setminus A$)? If $A$  were a regular macroscopic region, one would expect $|\partial A|$ to be of order $L^{d-1}$. The physics literature predicts instead that droplet interfaces in spin glasses are rougher, with an effective dimension $d_s>d-1$ \cite{fisherhuse86, braymoore87}. The following theorem, proved in Section \ref{fractalthmpf}, gives a rigorous lower bound of this type for droplets produced by perturbations of the first kind. We do not impose connectedness of the droplet; in particular, the perturbation droplet may be disconnected.

\begin{thm}\label{fractalthm}
Consider the EA model on $\{0,1,\ldots,L\}^d$ with the usual boundary and some given boundary condition, or with free or periodic boundary. Take any $p\in (0,1)$ and let $A$ be the region of overturned spins when the environment is given a perturbation of the first kind with parameter $p$. Then 
\[
\pp(|\partial A|\ge C_1 (1-p)\sqrt{p}L^d (\log L)^{-\frac{1}{2}}) \ge 1- 3e^{-C_2pL^d} - C_3L^{-C_4},
\]
where $C_1$, $C_2$, $C_3$ and $C_4$ are positive constants depending only on $d$. 
\end{thm}
Note that under free or periodic boundary conditions, $A$ itself is not canonical, but $\partial A$ is well-defined. Consequently, if $p=L^{-\alpha+o(1)}$, then $|\partial A|$ is at least $L^{d-\alpha/2+o(1)}$, up to logarithmic factors, with high probability. In the regime $p\gg L^{-1}$, where Corollary~\ref{chaoscor} implies that the droplet has macroscopic volume, this yields an interface lower bound of order at least $L^{d-1/2}$, up to logarithmic factors. Thus the interface dimension is forced to be strictly larger than $d-1$ in the macroscopic-droplet regime. Even for smaller perturbations, the theorem shows that once $p\gg L^{-d}$, the droplet boundary is already polynomially large, of order at least $L^{d/2}$  up to logarithmic factors.

\subsection{Existence of multiple valleys}
Our next result concerns low energy macroscopic excitations. In a ferromagnet, overturning a set of spins in the ground state costs energy proportional to the size of the edge boundary. In the EA model, by contrast, the droplet picture predicts the existence of macroscopic regions whose energy cost is much smaller than the size of their boundary \cite{fisherhuse88}. This prediction is also one of the heuristic inputs behind disorder chaos \cite{braymoore87}.

We restrict attention to subsets of $V^\circ$ whose sizes lie between $\frac14|V^\circ|$ and $\frac34|V^\circ|$. Given a region $A\subset V^\circ$, let $\Delta(A)$ denote the energy cost of  overturning all spins in $A$ in the ground state (and keeping all other spins the same). We are interested in showing that there is some set $A$ with $\frac{1}{4}|V^\circ|\le |A|\le \frac{3}{4}|V^\circ|$ such that the ratio $\frac{\Delta(A)}{|\partial A|}$ is small, where $\partial A$ is the edge-boundary of $A$ --- that is, the set of all edges from $A$ to $V\setminus A$.  (Note that $\partial A$ is nonempty because of the bounds on the size of $A$.) To do this, let us define
\[
F := \min\biggl\{\frac{\Delta(A)}{|\partial A|} : A\subseteq V^\circ, \frac{|V^\circ|}{4}\le |A|\le \frac{3|V^\circ|}{4}\biggr\}.
\]
The following result shows that $F$ is small with high probability whenever $|V^\circ|$ and $\frac{|V^\circ|}{|B|}$ are larger than some power of $\log |E|$. As before, we first state the general result, and then specialize to the case of $V= \{0,1,\ldots, L\}^d$ in the corollary that follows. The proof is in Section \ref{boundarythmpf}.
\begin{thm}\label{boundarythm}
Let all notations be as in Theorem \ref{mainthm}, and let $F$ be defined as above. Then there is a positive universal constant $C_1$ and a positive constant $C_2$ depending only $\alpha$, $\beta$, $\gamma$ and $\delta$, such that for any $p\in (0,\frac{1}{2})$, we have
\[
\pp(F \ge C_1 \sqrt{p\log |E|}) \le \frac{1}{|E|} + \frac{C_2 (|V^\circ| p^{-\beta} + |B|^2 p^{-2\delta})}{|V^\circ|^2}.
\]
\end{thm}
Recall that if $V= \{0,1,\ldots, L\}^d$ with the usual boundary, then $|V^\circ|$ is of order $L^d$, $|B|$ is of order $L^{d-1}$, $\beta = d$, and $\delta =1$. Additionally, note that $|E|$ is of order $L^d$. Taking $p = KL^{-1}$ for some fixed $K\ge 1$ gives the following corollary, which shows that $F$ is at most of order $L^{-\frac{1}{2}}\sqrt{\log L}$ in any dimension.
\begin{cor}\label{overcor}
In the setting of Theorem \ref{boundarythm}, if $V= \{0,1,\ldots, L\}^d$ with the usual boundary and with any given disorder-independent boundary condition, then for any $K\ge 1$, 
\[
\pp(F\ge C_1 L^{-\frac{1}{2}}\sqrt{K \log L}) \le C_2 K^{-2}.
\]
where $C_1, C_2$ are positive constants that depend only on $d$. For free or periodic boundary, the $K^{-2}$ on the right improves to $K^{-d}$ for $d\ge 3$. 
\end{cor}

In dimension one, a simple argument shows that the bound above is not sharp. With high probability, one can find two edges $e$ and $f$, separated by distance of order $L$, such that both $J_e$ and $J_f$ are of order $L^{-1}$. Overturning all spins between $e$ and $f$ then produces a region of size of order $L$ with energy cost of order $L^{-1}$. Thus, in dimension one, $\E(F)$ is at most of order $L^{-1}$. The bound in Corollary~\ref{overcor} may be suboptimal in higher dimensions as well, but the correct order for $d\ge 2$ remains open.

The physics literature is not unanimous about the size of $F$. For example, there are competing claims, made via numerical studies, that in $d=3$, the energy cost $\Delta(A)$  can be as small as $O(L^{\frac{1}{5}})$~\cite{braymoore87}, or $O(1)$~\cite{krzakalamartin00}. More generally, the physics literature (e.g., \cite{braymoore87, fisherhuse88}) indicates that there are regions where $\frac{\Delta(A)}{\sqrt{|\partial A|}}$ is small, and not just  $\frac{\Delta(A)}{|\partial A|}$. This seems inaccessible by the methods of this paper. We leave it as an open question.

Note that for a macroscopic region $A$, $|\partial A|$ is at least of order $L^2$ in $d=3$. The main difficulty with simulation studies is that finding the ground state is an NP-hard problem, with no good algorithm even for the ``average case''. Simulations can be carried out with only rather small values of $L$ (e.g., $L=12$ in \cite{krzakalamartin00}).

As a counterpart of Theorem \ref{boundarythm}, the next result shows that large regions with small interface energies, whose existence is guaranteed by Theorem \ref{boundarythm}, are exceptionally rare. The probability that any given region has a small interface energy is exponentially small in the size of the boundary. This is the content of the next theorem, proved in Section~\ref{oldthmpf}.
\begin{thm}\label{oldthm}
In the setting of Theorem \ref{mainthm}, there are positive constants $C_1$, $C_2$ and $C_3$ depending only on  the maximum degree of $G$, such that for any $A\subset V^\circ$,
\[
\pp\biggl(\frac{\Delta(A)}{|\partial A|} < C_1\biggr) \le C_2e^{-C_3|\partial A|}. 
\]
\end{thm}
This shows that the optimal $A$ in the definition of $F$ cannot be any given region, but rather, one that arises ``at random''. Indeed, in our proof of Theorem \ref{boundarythm}, an $A$ with a small value of $\frac{\Delta(A)}{|\partial A|}$ is obtained as the region of overturned spins after applying a small perturbation of the first kind.

\subsection{Size of the critical droplet}
The next result concerns the size of the so-called ``critical droplet'' of an edge, an object that has attracted some recent attention~\cite{arguinetal21}, where the critical droplets in two related models --- the ``highly disordered model'' and the ``strongly disordered model'' --- were studied and shown to be finite in the infinite volume limit. The critical droplet in our context is defined as follows. Take any edge $e = \{i,j\}$. Let $\sigma^1$ be the energy minimizing configuration under the constraint that $\sigma_i = \sigma_j$, and let $\sigma^2$ be the energy minimizing configuration under the constraint that $\sigma_i = -\sigma_j$. It is easy to see that $\sigma^1$ and $\sigma^2$ do not depend on the value of $J_e$, and for any value of $J_e$ (keeping all other fixed), the ground state of the system is either $\sigma^1$ or $\sigma^2$. The critical droplet is the set of sites where $\sigma^1$ and $\sigma^2$ disagree. Under the free or periodic boundary conditions on $G$, this is not completely well-defined, because if a set $A$ fits the above definition, then so does $V \setminus A$. In this case we define the size of the critical droplet (which is our main object of interest) as the minimum of $|A|$ and $|V \setminus A|$.

Let $D(e)$ be the critical droplet of an edge $e$, in the setting of Theorem \ref{mainthm}. Note that unlike the ``droplet'' $A$ in Theorem \ref{fractalthm}, it is not hard to show that the critical droplet $D(e)$ is a connected set. The following theorem, proved in Section \ref{dropletthmpf},  gives a lower bound on the expected value of the size of $D(e)$. 
\begin{thm}\label{dropletthm}
Let $D(e)$ be as above. Then, in the setting of Theorem \ref{mainthm}, 
\[
\frac{1}{|E|} \sum_{e\in E} \ee|D(e)| \ge \frac{C|V^\circ|}{|E| \max\bigl\{|V^\circ|^{-\frac{1}{\beta}}, \bigl(\frac{|B|}{|V^\circ|}\bigr)^{\frac{1}{\delta}}\bigr\}},
\]
where $C$ is a positive constant depending only on $\alpha$, $\beta$, $\gamma$ and $\delta$.
\end{thm}
Specializing to the case $V= \{0,1,\ldots,L\}^d$ with the usual boundary (or with free or periodic boundary), where $|V^\circ|$ and $|E|$ are of order $L^d$, $|B|$ is of order $L^{d-1}$, $\beta=d$ and $\delta = 1$, we obtain the following corollary. 
\begin{cor}
If $V= \{0,1,\ldots,L\}^d$ with the usual boundary and with any given disorder-independent boundary condition (or with free or periodic boundary), then 
\[
\frac{1}{|E|} \sum_{e\in E} \ee|D(e)| \ge C(d) L
\]
for all $d\ge 1$, where $C(d)$ is a positive constant depending only on $d$.
\end{cor}
In particular, under the periodic boundary condition on $V = \{0,1,\ldots,L\}^d$, we have $\ee|D(e)|\ge C(d)L$ for any $e$. This has the following consequence for $d\ge 2$. Since $|D(e)|\le \frac{1}{2}|V| = \frac{1}{2}L^d$ (due to the periodic boundary condition), an isoperimetric inequality of \citet[Theorem 8]{bollobasleader91} implies that
\begin{align*}
|\partial D(e)| &\ge \min_{1\le r\le d} 2|D(e)|^{1-\frac{1}{r}} rL^{\frac{d}{r} -1}\\
&\ge \min_{1\le r\le d} 2|D(e)|^{1-\frac{1}{r}} r(2|D(e)|)^{\frac{\frac{d}{r} -1}{d}}\ge 2 |D(e)|^{1-\frac{1}{d}}.
\end{align*}
Thus, we get the following corollary.
\begin{cor}
Take any $d\ge 2$ and $L\ge 1$. For $V = \{0,1,\ldots,L\}^d$ with periodic boundary condition, we have that for any edge $e$,
\[
\ee|\partial D(e)|^{\frac{d}{d-1}} \ge C(d) L,
\]
where $C(d)$ is a positive constant depending only on $d$.
\end{cor}
%Note that $\partial D(e)$ is the set of edge-spins that are overturned when $J_e$ is replaced by an independent copy, where we define the spin associated with an edge to be the product of the spins at its endpoints. This indicates (but does not prove) that edge-spins are also chaotic with respect to small perturbations of the disorder. 

 %What it does prove is that for an edge $e=\{i,j\}$, the dependence of $\sigma_i \sigma_j$ on $J_f$ can decrease at most polynomially in the distance between $e$ and $f$ (if it decays at all). This is made more precise by the next theorem. %This is in contrast to the one-dimensional situation, where the decay is exponential, as can be verified by explicitly writing down the ground state as a function of the disorder and the boundary condition. 

\subsection{Polynomial decay of correlations}
Our final result is about the decay of correlations in the ground state. This is slightly tricky to define, for the following reasons. First, it is not hard to see that if we integrate out the disorder, then the spins are all mutually independent and identically distributed. On the other hand, if we fix the disorder, then the spins become deterministic (under a given boundary condition). Thus, the only sensible way to understand the decay of correlations is to look at the dependence of the spin at the center of a cube on the boundary condition, after fixing the disorder. However, it is trivial to see that the spin at a site is heavily dependent on the boundary condition, since overturning all the spins on the boundary also overturns the spin at any given site. A nontrivial question emerges only if we consider the dependence of $\sigma_i\sigma_j$ on the boundary, where $i$ and $j$ are {\it neighboring} sites, and $\sigma$ is the ground state. 

The next theorem, proved in Section \ref{decaythmpf},  shows that for neighboring sites $i$ and $j$, the dependence of $\sigma_i\sigma_j$ on the boundary condition decreases at most like the inverse of the distance to the boundary.

\begin{thm}\label{decaythm}
Take $d\ge 2$ and $L\ge 2$. Let $V\subseteq \Z^d$ be the nearest-neighbor box of side length $L$ centered at the origin, and let $B$ be its vertex boundary. For each boundary condition $\gamma\in\{-1,1\}^{B}$, let $\sigma^\gamma$ be the corresponding ground state. Fix a neighbor $j$ of the origin, and let $\me$ be the event that the value of $\sigma^\gamma_0\sigma^\gamma_j$ is not the same for all choices of $\gamma$. Then
\[
    \P(\me)\ge \frac{1}{32L}.
\]
\end{thm}

This result contrasts with the recent proof of exponential decay of correlations in the two-dimensional random field Ising model (RFIM) by \citet{dingxia21}. The RFIM is a disordered system with Hamiltonian
\[
    -\sum_{\{i,j\}\in E}\sigma_i\sigma_j-\sum_{i\in V}J_i\sigma_i,
\]
where the $J_i$'s are i.i.d.~symmetric random variables, often taken to be Gaussian. A major technical advantage of the RFIM, in comparison with the EA model, is the FKG inequality. Ding and Xia showed, among other things, that if $E$ denotes the event that the ground state spin at the center of a box is not the same for all boundary conditions, then $\P(E)$ decays exponentially in the side length of the box in dimension two. Theorem~\ref{decaythm} shows that the analogous boundary-condition sensitivity for a nearest-neighbor spin product in the EA model cannot decay faster than order $L^{-1}$. This highlights a qualitative difference between the two models. In dimensions $d\ge 3$, the RFIM exhibits long-range order \cite{dingzhuang21, imbrie85}; it remains open whether an analogous statement holds for the EA model.

\subsection{Proof idea} 
A common mechanism underlies the proofs of these results. The key input is a spectral observation about ground state spin products. If $\sigma_x\sigma_y$ is expanded, as a function of the Gaussian disorder, in the Hermite basis, then every nonzero term in the expansion involves a set of edges that either connects $x$ to $y$, or connects both $x$ and $y$ to the boundary. Consequently, when $x$ and $y$ are far apart and far from the boundary, the low degree part of this expansion vanishes. Small Ornstein--Uhlenbeck perturbations damp high degree Hermite components, which yields the disorder chaos estimates. The droplet, multiple valley, and critical droplet results are then obtained by applying these estimates to the set of spins changed by the perturbation.

\subsection{Related literature and open problems}
The phenomenon of chaos in lattice spin glasses was proposed in the physics literature by \citet{fisherhuse86} and \citet{braymoore87}. Disorder chaos for the SK model was proved in \cite{chatterjee09, chatterjee14}. In \cite{chatterjee09}, it was also shown that the bond overlap in the EA model is not chaotic, in the sense that its value does not drop to zero under a small perturbation. This still leaves open the possibility that it drops to a nonzero value strictly less than the value at zero perturbation, which is the stronger definition of chaos for the bond overlap.  

To be more precise, let $\sigma$ and $\sigma'$ be the ground states in the original and perturbed environments, as in Theorem \ref{chaosthm}. Then Theorem \ref{chaosthm} shows  that $\ee(\sigma_i\sigma_j \sigma_i'\sigma_j')$ drops sharply to zero as the perturbation parameter $p$ increases from $0$ to a small positive value, if $i$ and $j$ are far apart. Now suppose $i$ and $j$ are neighbors. Then it was shown in \cite{chatterjee09, chatterjee14} that $\ee(\sigma_i\sigma_j \sigma_i'\sigma_j')$ will not drop to zero --- but does it drop sharply to a value less than $1$? That is, is it true that under the  periodic boundary condition on $\{-L,\ldots,L\}^d$, for fixed neighboring sites $i$ and $j$,
\[
\lim_{p\to 0} \lim_{L\to\infty} |\ee(\sigma_i\sigma_j \sigma_i'\sigma_j')| < 1?
\]
An answer to the above question will tell us whether the bond overlap in the EA model is chaotic (in the stronger sense) with respect to small perturbations in the disorder, or not.

Further investigations of  disorder chaos in the mean-field setting were carried out in \cite{chen13, chen14, chenpanchenko18, chenetal15, chensen17, auffingerchen16, chenetal18, chenetal17, eldan20}.  The related notion of temperature chaos in mean-field models was investigated in \cite{panchenko16, chenpanchenko17, benarousetal20, subag17}. Connections with computational complexity were explored in \cite{gamarnik21, huangsellke21}, and with noise sensitivity in \cite{garbansteif14}. 

In the lattice setting, there are fewer results, reflecting the general dearth of rigorous results for short-range models. The absence of disorder chaos in the bond overlap of the EA model, in the narrow sense that was proved for Gaussian couplings in \cite{chatterjee09, chatterjee14}, has been recently generalized to non-Gaussian couplings by  \citet{arguinhanson20}. A very interesting connection between disorder chaos in the bond overlap and the presence of incongruent states (explained below) was proved by \citet*{arguinetal19}, who showed that if there is no disorder chaos (in the stronger sense), then incongruent ground states cannot exist, at least as limits of finite volume ground states with disorder-independent boundary conditions. 

The problem of incongruent ground states is one of the central open problems for short-range spin glasses. The problem is stated most clearly in the infinite volume setting. Consider the EA Hamiltonian on the whole of $\zz^d$ instead of a finite region. The notion of ``minimizing the energy'' no longer makes sense, but the difference between the energies of two states that differ only at a finite number of sites is well-defined and finite. An infinite volume state is called a ground state if overturning any finite number of spins results in an increase in the energy. It was shown  by \citet{newmanstein92} that the number of ground states is almost surely equal to a constant depending on the dimension and the distribution of the disorder. The main question is whether this number is greater than two in some dimension and for some symmetric and continuous distribution of the disorder. Obviously, if $\sigma$ is an infinite volume ground state, then so is $-\sigma$, and so the number of ground states is at least two. Two ground states that are not related in this way are called incongruent ground states. The above question is the same as asking whether there can exist a pair of incongruent ground states. This question remains unanswered. The greatest progress on this topic  was made by \citet{newmanstein01}, who showed that in dimension two, if there is a pair of incongruent ground states, then there is a single doubly infinite ``domain wall'' dividing them. This result was used by \citet*{arguinetal10} to prove that there is a unique infinite volume ground state in the EA model on the half-plane  $\zz\times \{0,1,\ldots\}$ under a certain sequence of boundary conditions. The boundary condition was later eliminated by~\citet{arguindamron14}, who showed the number of ground state  pairs for the EA model on the half-plane is either $1$ or $\infty$. A related result by~\citet{bergertessler17} shows that for the ground state of the EA model on $\zz^2$, ``unsatisfied edges'' (i.e., where $\sigma_i\sigma_j\ne \mathrm{sign}(J_{ij})$) do not percolate. 

The absence of incongruent infinite volume ground states, if true, will have the following consequence in finite volume. Any two states that are nearly energy-minimizing will locally look like a pair of congruent states (i.e., either equal or negations of each other), although they may be globally quite different. The ``almost orthogonal'' states produced by the small perturbations of the disorder in Theorem \ref{mainthm} may (or may not) have this property. In the physics literature, this is known as ``regional congruence''~\cite{husefisher87}.

An important contribution to the study of the EA model from the mathematical literature is the concept of metastates, introduced by \citet{aizenmanwehr89, aizenmanwehr90}. A zero-temperature metastate is a measurable map taking the disorder in infinite volume to a probability measure on the set of ground states. \citet{aizenmanwehr89, aizenmanwehr90} showed that metastates exist, and an explicit construction and interpretation was given later by \citet{newmanstein98}. Metastates capture some aspects of the chaotic nature of spin glasses, such as the ``chaotic size dependence'' proved in \cite{newmanstein92}, which means that the ground state in a finite region is chaotic with respect to changes in the size of the region. For some recent results and different perspectives on metastates, see~\cite{cotaretal18}.

Another topic that has received considerable attention in the mathematical literature is the question of fluctuations of the ground state energy (and more generally, the free energy at any temperature). This began with the aforementioned papers of \citet{aizenmanwehr89, aizenmanwehr90}, who showed that the fluctuations are of the same order as the volume of the system. The motivation for studying fluctuations is that one can connect it to the question of phase transitions via the ``Imry--Ma argument''~\cite{imryma75}. This was made precise in \cite{aizenmanwehr89, aizenmanwehr90}. For further developments in the study of fluctuations, and especially the important topic of interface energy fluctuations, we refer to \cite{arguinetal14, arguinetal16, stein16, darioetal21} and references therein.

Besides the above, the few other results that have been proved rigorously about the EA model include stochastic stability and the Ghirlanda--Guerra identities in a weak form for the edge overlap. We refer to the excellent monograph by \citet{contuccigiardina13} for details.

%Besides the above, the few other results that have been proved rigorously about the EA model include stochastic stability and the Ghirlanda--Guerra identities. We refer to the excellent monograph by \citet{contuccigiardina13} for details.

One of the great unsolved questions in spin glass theory concerns the validity of the ``Parisi picture''~\cite{parisi07} versus the ``droplet theory'' of \citet{fisherhuse86}. Conclusively settling this controversy has remained out of the reach of rigorous mathematics to this day. Theorem~\ref{boundarythm} and Corollary~\ref{overcor} in the present paper are related to this problem. As explained by \citet{krzakalamartin00}, the Parisi picture implies that one can overturn all spins in a macroscopic subset of $\{0,1,\ldots,L\}^3$ with $O(1)$ energy cost, whereas the droplet theory implies that the minimum cost grows as a small positive power of $L$. While Corollary \ref{overcor} does not settle this debate, it is the first result to show that one can indeed find macroscopic regions with interface energies that are negligible compared to the size of the interface.

%This, too, would shed light on the debate between the Parisi picture versus the droplet theory. The Parisi picture holds if the above is true; the droplet theory, if not. Moreover, it would settle the longstanding question about the existence of incongruent ground states. Incongruent ground states exist if and only if the above inequality holds. 
An interesting question about the bond overlap and the site overlap is whether there is a relation between these two quantities. The property that they approximately behave like functions of each other at large volumes was introduced in \citet{parisiricci-tersenghi00} and numerically verified by \citet*{contuccietal06} and further developed in \citet*{contuccietal07}. This is known as ``overlap equivalence''. The overlap equivalence picture contradicts the ``trivial-non-trivial'' (TNT) picture proposed by \citet{palassiniyoung00} and \citet{krzakalamartin00}, which says that at low temperatures, the bond overlap is concentrated near a deterministic value, but the site overlap fluctuates. In this context, note that although our Theorem \ref{mainthm} shows that the site overlap concentrates near zero after a small perturbation, it is unclear whether that contradicts the TNT picture, since it does not say anything about the site overlap in the absence of perturbation.

%An interesting question about the bond overlap and the site overlap is whether there is a relation between these two quantities. There is a conjecture, supported by numerical evidence, that these quantities approximately behave like functions of each other at low temperatures. This is known as `overlap equivalence', a notion introduced by \citet*{contuccietal06} and further developed in \citet*{contuccietal07}. 

\subsection*{Acknowledgements}
I thank Louis-Pierre Arguin, Pierluigi Contucci, Daniel Fisher and Jafar Jafarov for helpful comments, and Tim Sudijono for checking the proofs. This work was partially supported by NSF grants DMS-2113242 and DMS-2153654.

 %There are in fact no other rigorous results on this topic, although it is well-known among physicists.

%Lastly, one topic that deserves mention is another widely studied disordered system, called the `random field Ising model' (RFIM). In the RFIM, although

\section{Proofs}\label{proofsec}

\subsection{Proof of Theorem \ref{mainthm}}\label{mainthmpf}

%Consider the EA model on a finite simple graph $G = (V,E)$ with boundary $B\subseteq V$, which may be empty. Let $\Delta$ be the maximum degree of $G$ and $D$ be its diameter.  Let $J$ be the disorder, which we assume to be i.i.d.~standard Gaussian. Let $\sigma =\sigma(J)$ be the energy-minimizing configuration for a given realization of $J$. If $B=\emptyset$, then there are two such configurations (one being the negation of the other), but any even polynomial of the coordinates is the same for both configurations and therefore, uniquely defined. 

The proof of Theorem \ref{mainthm} is based on a spectral argument, reminiscent of the analysis of dynamical percolation in \cite{garbanetal10} and the proof of disorder chaos in the SK model in \cite{chatterjee14}. Let $h_0,h_1,\ldots$ be the orthonormal basis of normalized Hermite polynomials for $L^2(\mu)$, where $\mu$ is the standard Gaussian distribution on $\R$ and $h_0\equiv 1$. Then an orthonormal basis of $L^2(J)$ is formed by products like $h_n(J) := \prod_{e\in E} h_{n_{e}} (J_e)$, where $n_e\in \N := \{0,1,\ldots\}$ for each $e$, and $n := (n_e)_{e\in E}\in \N^E$. Any square-integrable function $f(J)$ of the disorder $J$ can be expanded in this basis as
\begin{align}\label{hermite}
f(J) = \sum_{n\in \N^E} \hat{f}(n) h_n(J),
\end{align}
where 
\[
\hat{f}(n) := \ee(f(J)h_n(J)). 
\]
The infinite series on the right side in \eqref{hermite} should be interpreted as the $L^2$-limit of partial sums, where the order of summation is irrelevant.

Now take any distinct $i,j\in V^\circ$. Let $\sigma$ be a ground state in the EA model on $G$ (with or without a boundary condition). Consider $\sigma_i \sigma_j$ as a function $\phi(J)$ of the disorder $J$. This function is well-defined even if we do not impose a boundary condition. Obviously, it is in $L^2(J)$. For  any $n\in \N^E$, let $E(n)$ be the set of edges $e\in E$ such that $n_e >0$, and $V(n)$ be the set of vertices that are endpoints of the edges in $E(n)$. Then $G(n) := (V(n), E(n))$ is a subgraph of~$G$. The following lemma is the main ingredient for  the proof of Theorem \ref{mainthm}. 
\begin{lmm}\label{sizelmm}
Let all notations be as above. Then $\hat{\phi}(n) = 0$ unless both $i$ and $j$ are in $V(n)$ and the connected components of $G(n)$ containing $i$ and $j$ are either the same, or both intersect $B$. 
\end{lmm}
\begin{proof}
%Suppose that $\hat{\phi}(n) = 0$. It suffices to show that $i\in V(n)$ and the connected component of $G(n)$ containing $i$ intersects $B\cup \{j\}$.
First, suppose that $i\in V(n)$. Let $A$ be the connected component of $G(n)$ that contains $i$. Suppose that $A\cap (B\cup \{j\}) =\emptyset$. Let $\partial A$ be the set of edges from $A$ to $A^c := V \setminus A$. Since $A$ is a connected component of $G(n)$, no edge in $\partial A$ can be a member of $E(n)$. Define a new environment $J'$ as 
\[
J'_e =
\begin{cases}
-J_e &\text{ if } e\in \partial A,\\
J_e &\text{ if } e\notin \partial A.
\end{cases}
\]
Note that $J$ and $J'$ have the same law, since the disorder distribution is symmetric around zero, and the disorder variables are independent. Since $\partial A \cap E(n)=\emptyset$, $h_n(J)$ does not depend on $(J_e)_{e\in \partial A}$. Thus, 
\begin{align}\label{phineq}
\hat{\phi}(n) &= \ee(\phi(J) h_n(J)) = \ee(\ee(\phi(J)|(J_e)_{e\notin\partial A}) h_n(J)). 
\end{align} 
But note that since $J$ and $J'$ have the same law, and $J_e' = J_e$ for $e\notin\partial A$, 
\begin{align}\label{phijeq}
\ee(\phi(J)|(J_e)_{e\notin\partial A}) &= \ee(\phi(J')|(J_e)_{e\notin\partial A}).
\end{align}
Now, let $\sigma'$ be the configuration defined as 
\[
\sigma'_k =
\begin{cases}
-\sigma_k &\text{ if } k\in  A,\\
\sigma_k &\text{ if } k\notin A.
\end{cases}
\]
Then $\sigma'$ satisfies the given boundary condition (if any) since $A\cap B = \emptyset$. Let us now split $H_{J'}(\sigma')$ as 
\begin{align*}
H_{J'}(\sigma') &= -\sum_{\substack{\{k,l\}\in E, \\  k,l\in A}}J_{kl}(- \sigma_k )(-\sigma_l)  - \sum_{\{k,l\}\in \partial A} (-J_{kl}) (-\sigma_k \sigma_l) - \sum_{\substack{\{k,l\}\in E,\\ k,l\in A^c}} J_{kl} \sigma_k \sigma_l\\
&= -\sum_{\{k,l\}\in E} J_{kl} \sigma_k \sigma_l = H_{J}(\sigma). 
\end{align*}
Moreover, for any $\tau\in \{-1,1\}^V$ satisfying the given boundary condition (if any), we have $H_{J'}(\tau) = H_J(\tau')$, where $\tau'_i = -\tau_i$ if $i\in A$ and $\tau'_i = \tau_i$ if $i\notin A$. Since $\tau'$ also satisfies the given boundary condition, this shows that $\sigma'$ minimizes $H_{J'}$, and so $\sigma'_i \sigma_j' = \phi(J')$. But since $j\notin A$ and $i\in A$, $\sigma'_i \sigma'_j = -\sigma_i\sigma_j$. Thus, $\phi(J')=-\phi(J)$, and so, by \eqref{phijeq},
\[
\ee(\phi(J)|(J_e)_{e\notin\partial A}) = 0.
\]
Plugging this into \eqref{phineq}, we get that $\hat{\phi}(n) = 0$ if $i\in V(n)$ and $A$ does not intersect $B\cup \{j\}$. %Similarly, $\hat{\phi}(n) = 0$ if $j\in V(n)$ and the connected component of $G(n)$ containing $j$ does not intersect $B\cup \{i\}$. 

%Finally, suppose that $V(n)$ contains neither $i$ nor $j$. 

Next, suppose that $i\notin V(n)$. In this case, taking $A := \{i\}$ and repeating the whole argument as above shows that $\hat{\phi}(n) = 0$. Thus, $\hat{\phi}(n) = 0$ unless $i\in V(n)$ and $A$ intersects $B\cup \{j\}$.

By the symmetry between $i$ and $j$, we conclude that  $\hat{\phi}(n) = 0$ unless $j\in V(n)$ and the component of $G(n)$ containing $j$ intersects $B\cup \{i\}$. Combining these two conclusions yields the claim of the lemma.
\end{proof}
Lemma \ref{sizelmm} gives the following key corollary, which says that if $i$ and $j$ are far apart and far away from the boundary, then the Hermite polynomial expansion of $\sigma_i \sigma_j$ consists of only high degree terms. 
\begin{cor}\label{phicor}
If $\hat{\phi}(n)\ne 0$, then $|E(n)|\ge \min\{d(i,j), d(i,B)+d(j,B)\}$, where $d$ is the graph distance on $G$ and $d(i,B) := \min\{d(i,k):k\in B\}$ (which is infinity if $B$ is empty).  
\end{cor}
\begin{proof}
Suppose that $\hat{\phi}(n)\ne 0$. Then by Lemma \ref{sizelmm}, $i,j\in V(n)$ and the connected components containing $i$ and $j$ are either the same, or they are distinct and they both touch $B$. In the first case, there is a path of edges in $E(n)$ connecting $i$ to $j$, which implies that $|E(n)|\ge d(i,j)$. In the second case, there is a path in $G(n)$ connecting $i$ to $B$ and another path in $G(n)$ connecting $j$ to $B$, which implies that $|E(n)|\ge d(i,B)+d(j,B)$. 
\end{proof}
In the following, instead of using the parameter $p$ for the perturbation, we will reparametrize $p$ as $1-e^{-t}$, where $t\in (0,\infty)$. This is helpful for the following reason. Let 
\[
J(t) := e^{-t} J + \sqrt{1-e^{-2t}} J' = (1-p) J + \sqrt{2p - p^2} J'.
\]
Then $J(t)$ is the perturbed environment for our first kind of perturbation. It is a standard fact that for any $f\in L^2(J)$, $\ee(f(J(t))|J) = P_t f(J)$, where $(P_t)_{t\ge 0}$ is the Ornstein--Uhlenbeck semigroup (see, e.g., \cite[Chapter 2 and Chapter 6]{chatterjee14}). Moreover, for each $n\in \N^E$, $h_n$ is an eigenfunction of the Ornstein--Uhlenbeck generator, with eigenvalue 
\[
|n| := \sum_{e\in E} n_e.
\]
This implies that for any $f\in L^2(J)$,
\begin{align*}
P_t f (J) = \sum_{n\in \N^E} e^{-|n|t} \hat{f}(n) h_n(J). 
\end{align*}
In particular, by the Parseval identity,
\begin{align}\label{pars1}
\ee[(\ee(f(J(t))|J))^2] = \|P_t f(J)\|_{L^2}^2 &= \sum_{n\in \N^E} e^{-2|n|t} \hat{f}(n)^2. 
\end{align}
Now consider the second kind of perturbation, where each $J_e$ is replaced by an independent copy $J_e'$ with probability $p$. Let us again reparametrize $p = 1-e^{-t}$ and set $J(t)$ to be the new environment produced by the perturbation. Recall that $h_0\equiv 1$, and for $n\in \N \setminus \{0\}$, $h_n$ integrates to zero under the standard Gaussian measure on $\R$. This implies that for any $n\in \N^E$, 
\[
\ee(h_n(J(t))|J) = (1-p)^{\delta(n)} h_n(J) = e^{-\delta(n) t} h_n(J),
\]
where 
\[
\delta(n) := |\{e\in E: n_e > 0\}|. 
\]
Therefore, in this case,
\begin{align*}
\ee(f(J(t))|J) &= \sum_{n\in \N^E} e^{-\delta(n)t} \hat{f}(n) h_n(J)
\end{align*}
and hence,
\begin{align}\label{pars2}
\ee[(\ee(f(J(t))|J))^2] &= \sum_{n\in \N^E} e^{-2\delta(n)t} \hat{f}(n)^2. 
\end{align}
Combining the above observations with Corollary \ref{phicor}, we get the following lemma. 
\begin{lmm}\label{philmm}
Let $\sigma(t)$ be a ground state for the perturbed environment $J(t)$, where the perturbation is either of the two kinds described above. Then for any distinct $i,j\in V^\circ$, 
\[
\ee[(\ee(\sigma_i(t)\sigma_j(t)|J))^2] \le e^{-2t\min\{d(i,j), d(i,B)+d(j,B)\}}.
\]
\end{lmm}
\begin{proof}
Since $\delta(n)\le |n|$, the inequalities \eqref{pars1} and \eqref{pars2} shows that for either kind of perturbation, 
\begin{align*}
\ee[(\ee(\sigma_i(t)\sigma_j(t)|J))^2] &\le \sum_{n\in \N^E} e^{-2\delta(n) t} \hat{\phi}(n)^2.
\end{align*}
By Corollary \ref{phicor}, we know that $\hat{\phi}(n)=0$ unless $\delta (n)\ge \min\{d(i,j), d(i,B)+d(j,B)\}$. Combining this with the above inequality completes the proof.
\end{proof}

We are now ready to prove Theorem \ref{mainthm}.
\begin{proof}[Proof of Theorem \ref{mainthm}]
We will work with the reparametrization $p = 1-e^{-t}$, and write $R(t)$ instead of $R(p)$. Thus, by Lemma \ref{philmm}, 
\begin{align}
\ee(R(t)^2) &= \frac{1}{|V^\circ|^2} \sum_{i, j\in V^\circ}\ee(\sigma_i \sigma_j \sigma_i(t)\sigma_j(t))\notag \\
&=  \frac{1}{|V^\circ|^2} \sum_{i, j\in V^\circ}\ee(\sigma_i \sigma_j \ee(\sigma_i(t)\sigma_j(t)|J))\notag \\
&\le  \frac{1}{|V^\circ|^2} \sum_{i, j\in V^\circ}\ee|\ee(\sigma_i(t)\sigma_j(t)|J)|\notag \\ 
&\le \frac{1}{|V^\circ|^2} \sum_{i, j\in V^\circ}\sqrt{\ee[(\ee(\sigma_i(t)\sigma_j(t)|J))^2]}\notag\\
&\le  \frac{1}{|V^\circ|} + \frac{1}{|V^\circ|^2} \sum_{\substack{i, j\in V^\circ,\\ i\ne j}} e^{-t\min\{d(i,j), d(i,B)+d(j,B)\}}. \label{neweq}
\end{align}
For each $k\in \N$, let
\[
N_k := \biggl|\biggl\{(i,j): i,j\in V^\circ, i\ne j, \frac{k}{t} \le \min\{d(i,j), d(i,B) + d(j,B)\} \le \frac{k+1}{t}\biggr\}\biggr|. 
\]
Then note that
\begin{align}\label{eteq}
 \sum_{\substack{i, j\in V^\circ,\\ i\ne j}} e^{-t\min\{d(i,j), d(i,B)+d(j,B)\}} &\le \sum_{k=0}^\infty N_ke^{-k}. 
\end{align}
%Since the left side is decreasing in $t$, it is not hard to see that it suffices to prove the theorem under the assumption that $t\in (0,1)$. 
%Let us first complete the proof under the assumption that $t\in (0,1)$. 
By the given conditions, we have that for any $i\in V^\circ$ and $k\in \N$, the number of $j\in V^\circ\setminus\{i\}$ such that $d(i,j)\le \frac{k+1}{t}$ is at most $\alpha (k+1)^\beta t^{-\beta}$. Note that this holds even if $\frac{k+1}{t} < 1$, since in that case the number is zero. Thus, the number of pairs $(i,j)$ such that $i\ne j$ and $d(i,j)\le \frac{k+1}{t}$ is at most $\alpha |V^\circ| (k+1)^\beta t^{-\beta}$. Similarly, note that the number of $i\in V^\circ$ such that $d(i,B)\le \frac{k+1}{t}$ is at most $|B| \gamma (k+1)^\delta t^{-\delta}$. Thus, the number of pairs $(i,j)$ such that $i,j\in V^\circ$, $i\ne j$, and $d(i,B)+d(j,B)\le \frac{k+1}{t}$ is at most $|B|^2 \gamma^2 (k+1)^{2\delta} t^{-2\delta}$. Combining, we get 
\[
N_k \le \alpha |V^\circ| (k+1)^\beta t^{-\beta} + |B|^2 \gamma^2 (k+1)^{2\delta} t^{-2\delta}. 
\]
Plugging this bound into \eqref{eteq}, we get
\[
 \sum_{\substack{i, j\in V^\circ,\\ i\ne j}} e^{-t\min\{d(i,j), d(i,B)+d(j,B)\}}  \le C (|V^\circ| t^{-\beta} + |B|^2 t^{-2\delta}),
\]
where $C$ depends only on $\alpha$, $\beta$, $\gamma$ and $\delta$. Since $t = -\log(1-p)$, series expansion shows that $t\ge p$ for all $p\in (0,1)$. This completes the proof of Theorem \ref{mainthm}. 
%when $t\in (0,1)$. Finally, if $t\ge 1$, then by comparing with the case $t=\frac{1}{2}$, we get that
%\begin{align*}
% &\frac{1}{|V^\circ|} + \frac{1}{|V^\circ|^2} \sum_{\substack{i, j\in V^\circ,\\ i\ne j}} e^{-t\min\{d(i,j), d(i,B)+d(j,B)\}} \\&\le  \frac{1}{|V^\circ|} + \frac{1}{|V^\circ|^2} \sum_{\substack{i, j\in V^\circ,\\ i\ne j}} e^{-\frac{1}{2}\min\{d(i,j), d(i,B)+d(j,B)\}}\\
% &\le  \frac{1}{|V^\circ|} + \frac{C (|V^\circ| (1-e^{-\frac{1}{2}})^{-\beta} + |B|^2 (1-e^{-\frac{1}{2}})^{-2\delta})}{|V^\circ|^2}. 
%\end{align*}
%But $(1-e^{-\frac{1}{2}})^{-1} \le C''p^{-1}$ for some constant $C''$, since $p \ge 1-e^{-1}$. The proof is now completed by invoking the bound \eqref{neweq}.
\end{proof}

\subsection{Proof of Theorem \ref{chaosthm}}\label{chaosthmpf}
Let us reparametrize $p=1-e^{-t}$, as before. Let $\sigma' := \sigma(t)$, in the notation of Lemma \ref{philmm}. Then by Lemma \ref{philmm},
\begin{align*}
|\ee(\sigma_i\sigma_j \sigma_i'\sigma_j')| &= |\ee(\sigma_i\sigma_j \ee(\sigma_i'\sigma_j'|J))|\\
&\le \ee|\ee(\sigma_i'\sigma_j'|J)|\\
&\le \sqrt{\ee[(\ee(\sigma_i'\sigma_j'|J))^2]}\\
&\le e^{-t\min\{d(i,j), d(i,B)+d(j,B)\}}. 
\end{align*}
Since $e^{-t} = 1-p$, this completes the proof of Theorem \ref{chaosthm}. 

\subsection{Proof of Theorem \ref{boundarythm}}\label{boundarythmpf}
Take any $p\in (0,\frac{1}{2})$. Let $J(p)$ be the perturbed Hamiltonian, with the first kind of perturbation --- that is, $J(p) = (1-p)J + \sqrt{2p-p^2} J'$, where $J'$ is an independent copy of $J$. Let $\sigma(p)$ be the ground state for the perturbed environment. Let $A$ be the region where $\sigma(p)$ disagrees with $\sigma$. Then note that 
\begin{align*}
|A| &= |\{i\in V^\circ:  \sigma_i\sigma_i(p) = -1\}|\\
&= \frac{1}{2}(|\{i\in V^\circ:  \sigma_i\sigma_i(p) = -1\}| + |V^\circ| - |\{i\in V^\circ:  \sigma_i\sigma_i(p)  = 1\}|)\\
&= \frac{|V^\circ|}{2} - \frac{|V^{\circ}|R(p)}{2}. 
\end{align*}
Thus,
\begin{align}
\pp\biggl(\biggl||A|-\frac{|V^\circ|}{2}\biggr| > \frac{|V^\circ|}{4}\biggr) &= \pp\biggl(|R(p)|> \frac{1}{2}\biggr) \le 4 \ee(R(p)^2). \label{vbound}
\end{align}
Next, note that for adjacent $i,j$,
\begin{align}\label{heq1}
\sigma_i(p)\sigma_j(p) = 
\begin{cases}
- \sigma_i \sigma_j &\text{ if } \{i,j\}\in \partial A,\\
\sigma_i\sigma_j &\text{ otherwise.}
\end{cases}
\end{align}
Thus,
\begin{align*}
H_J(\sigma(p)) - H_J(\sigma) &= 2\sum_{\{i,j\}\in \partial A} J_{ij}\sigma_i \sigma_j.
\end{align*}
On the other hand, since $\sigma(p)$ minimizes $H_{J(p)}$,
\begin{align*}
H_{J(p)}(\sigma) - H_{J(p)}(\sigma(p)) \ge 0. 
\end{align*}
But note that by equation \eqref{heq1},
\begin{align*}
H_{J(p)}(\sigma) - H_{J(p)} (\sigma(p)) &= -2\sum_{\{i,j\}\in \partial A} J_{ij}(p)\sigma_i \sigma_j\\
&= -2(1-p)\sum_{\{i,j\}\in \partial A} J_{ij}\sigma_i \sigma_j - 2\sqrt{2p - p^2}\sum_{\{i,j\}\in \partial A} J_{ij}'\sigma_i \sigma_j\\
&= -(1-p) (H_J(\sigma(p)) - H_J(\sigma)) - 2\sqrt{2p - p^2}\sum_{\{i,j\}\in \partial A} J_{ij}'\sigma_i \sigma_j.
\end{align*}
Combining the last two displays, we get
\begin{align*}
H_J(\sigma(p)) - H_J(\sigma) &\le -\frac{2 \sqrt{2p-p^2}}{1-p} \sum_{\{i,j\}\in \partial A} J_{ij}'\sigma_i \sigma_j\\
&\le \frac{2 \sqrt{2p-p^2}}{1-p} |\partial A| \max_{\{i,j\}\in E} |J_{ij}'|. 
\end{align*}
But $H_J(\sigma(p)) - H_J(\sigma) = \Delta(A)$. Thus, 
\begin{align}\label{sbound}
\frac{\Delta(A)}{|\partial A|} &\le \frac{2 \sqrt{2p-p^2}}{1-p}\max_{\{i,j\}\in E} |J_{ij}'|.
\end{align}
Using a standard tail bound for Gaussian random variables, we have that for any $x\ge 0$,
\begin{align*}
\pp(\max_{\{i,j\}\in E} |J_{ij}'| \ge x) &\le \sum_{\{i,j\}\in E} \pp(|J_{ij}'|\ge x)\le 2|E| e^{-\frac{1}{2}x^2}.
\end{align*}
Combining this with \eqref{sbound}, \eqref{vbound}, and the fact that $p< \frac{1}{2}$, we get that for a sufficiently large universal constant $C$, 
\[
\pp(F \ge C\sqrt{p\log |E|}) \le |E|^{-1} + 4 \ee(R(p)^2).
\]
Invoking Theorem \ref{mainthm} to bound $\ee(R(p)^2)$ completes the proof.
%This completes the proof of Theorem \ref{boundarythm} under the assumption that $|V^\circ|$ and $|V^\circ|/|B|$ are larger than some constant depending on $\alpha$, $\beta$, $\gamma$ and $\delta$. To drop this assumption, simply note that if $p > 1/2$, then the above inequality holds trivially since the left side is bounded above by one. %To fully complete the proof, let us now show that this assumption can be dropped. Note that by \eqref{sineq}, $F$ is always bounded above by $2\max_{e\in E} |J_e|$. Thus, we always have that $\ee(F)\le 2\sqrt{4\log |E|}$. This shows that by sufficiently increasing the constant $C$ in the statement of the theorem, we can have the required inequality hold without the largeness assumption on $|V^\circ|$ and $|V^\circ|/|B|$. 

\subsection{Proof of Theorem \ref{fractalthm}}\label{fractalthmpf}
Throughout this proof, $C_1, C_2,\ldots$ will denote positive constants that depend only on $d$. We will continue to use the notations from the proof of Theorem \ref{boundarythm}. First, note that by~\eqref{sbound},
\begin{align}\label{partial}
|\partial A| &\ge \frac{(1-p)\Delta(A)}{2\sqrt{2p} M},
\end{align}
where
\[
M := \max_{\{i,j\}\in E} |J_{ij}'|.
\]
Next, note that
\begin{align}
\Delta(A) &= H_J(\sigma(p)) - H_J(\sigma)\notag \\
&= H_J(\sigma(p)) - H_{J(p)}(\sigma(p)) + H_{J(p)}(\sigma(p)) - H_J(\sigma). \label{deltaeq}
\end{align}
We will separately estimate the two terms on the right. First, let 
\[
K_{ij} := \sqrt{2p-p^2}J_{ij} - (1-p)J_{ij}',
\]
so that 
\begin{align*}
J_{ij} &= (1-p) J_{ij}(p) + \sqrt{2p-p^2}K_{ij}.
\end{align*}
Then we have
\begin{align}\label{h0}
&H_J(\sigma(p)) - H_{J(p)}(\sigma(p)) = - \sum_{\{i,j\}\in E} (J_{ij}-J_{ij}(p)) \sigma_i(p)\sigma_j(p)\notag \\
&= p\sum_{\{i,j\}\in E} J_{ij}(p)\sigma_i(p)\sigma_j(p) -\sqrt{2p-p^2}\sum_{\{i,j\}\in E} K_{ij} \sigma_i(p)\sigma_j(p). 
\end{align}
Now, it is easy to show (e.g., by Sudakov minoration~\cite[Lemma A.3]{chatterjee14} or otherwise) that 
\begin{align*}
\ee\biggl(\sum_{\{i,j\}\in E} J_{ij}(p)\sigma_i(p)\sigma_j(p) \biggr) \ge C_1 L^d. 
\end{align*}
Combining this with the concentration inequality for maxima of Gaussian fields~\cite[Equation (A.7)]{chatterjee14}, we get
\begin{align}\label{p1}
\pp\biggl( \sum_{\{i,j\}\in E} J_{ij}(p)\sigma_i(p)\sigma_j(p)\ge C_2 L^d\biggr) &\ge 1 - e^{-C_3L^d}. 
\end{align}
Now note that $\cov(J_{ij}(p), K_{ij}) = 0$, which implies that $J(p)$ and $K$ are independent. Moreover, $\var(K_{ij}) = 1$. Since $\sigma(p)$ is a function of $J(p)$, this shows that 
\[
\sum_{\{i,j\}\in E} K_{ij} \sigma_i(p)\sigma_j(p) \sim \mathcal{N}(0, |E|). 
\]
In particular,
\begin{align}\label{p2}
\pp\biggl(\sum_{\{i,j\}\in E} K_{ij} \sigma_i(p)\sigma_j(p) \ge \frac{C_2\sqrt{p}L^d}{4}\biggr) &\le  e^{-C_4 pL^d}. 
\end{align}
Combining \eqref{h0}, \eqref{p1} and \eqref{p2}, we get
\begin{align}\label{h1}
\pp(H_J(\sigma(p)) - H_{J(p)}(\sigma(p))\ge C_5pL^d) \ge 1-2e^{-C_6pL^d}. 
\end{align}
Next, let 
\begin{align*}
Q_{ij} &:= \sqrt{1-p}J_{ij} + \sqrt{p}J_{ij}',\\
R_{ij} &:= \sqrt{1-p} J_{ij} + \sqrt{p}J_{ij}'',
\end{align*}
where $J''$ is another independent copy of $J$. 
Then note that 
\begin{align*}
\cov(Q_{ij}, R_{ij}) &= 1-p = \cov(J_{ij}, J_{ij}(p)). 
\end{align*}
This shows that $(Q,R)$ has the same joint distribution as $(J, J(p))$. Thus, if $\sigma^1$ and $\sigma^2$ denote the ground states for $Q$ and $R$, then for any $x \ge 0$,
\begin{align*}
\pp(H_{J(p)}(\sigma(p)) - H_J(\sigma) \le -x) &= \pp(H_{R}(\sigma^2) - H_Q(\sigma^1) \le -x).
\end{align*}
Now, conditional on $J$, the random variables $Q$ and $R$ are independent and identically distributed. The conditional means are not zero, but the coordinates are independent and the conditional variance of each coordinate is $p$. In particular, $\ee(H_R(\sigma^2)-H_Q(\sigma^1)|J)=0$. Moreover, it is not hard to show that if $J$ is fixed, then as a function of the Gaussian random  vector $(J', J'')$, $H_R(\sigma^2) - H_Q(\sigma^1)$ is Lipschitz (with respect to Euclidean distance) with a Lipschitz constant bounded above by $\sqrt{2p |E|}$.  From this, an application of the Gaussian concentration inequality~\cite[Equation (A.5)]{chatterjee14} yields that for any $x\ge 0$, 
\begin{align*}
 \pp(H_{R}(\sigma^2) - H_Q(\sigma^1) \le -x|J) &\le \exp\biggl(-\frac{x^2}{4p|E|}\biggr). 
\end{align*}
Thus,
\begin{align}\label{h2}
\pp\biggl(H_{J(p)}(\sigma(p)) - H_J(\sigma) \le - \frac{C_5 pL^d}{2}\biggr) &\le e^{-C_7 pL^d}. 
\end{align}
Combining \eqref{deltaeq}, \eqref{h1} and \eqref{h2}, we get
\begin{align*}
\pp(\Delta(A) \ge C_8 pL^d) \ge 1- 3e^{-C_9pL^d}. 
\end{align*}
Finally, note that for any $K>0$,
\begin{align*}
\pp(M > K\sqrt{\log L}) &\le \sum_{\{i,j\}\in E}\pp(|J_{ij}'|> K\sqrt{\log L}) \\
&\le C_{10}L^de^{-\frac{1}{2}K^2\log L}= C_{10} L^{d - \frac{1}{2}K^2}. 
\end{align*}
Combining the last two displays with \eqref{partial} completes the proof. 

\subsection{Proof of Theorem \ref{oldthm}}\label{oldthmpf}
Take any edge $e = \{i,j\}\in \partial A$, where $A$ is now a given set as in the statement of Theorem \ref{oldthm}, and not random as in the previous subsection. Let $H_1$ (resp., $H_2$) be the minimum energy of the system subject to the constraints $\sigma_i = \sigma_j$ (resp., $\sigma_i = -\sigma_j$) and $J_e=0$, keeping all other edge weights intact. Then the ground state energy is $\min\{-J_e + H_1, J_e + H_2\}$. Moreover, the ground state satisfies $\sigma_i = \sigma_j$ if $-J_e +H_1 < J_e +H_2$ and $\sigma_i = -\sigma_j$ if $-J_e + H_1 > J_e + H_2$. (Note that these are the only possibilities, since equality occurs with probability zero.) The conditions can be rewritten as $J_e > \frac{1}{2}(H_1-H_2)$ and $J_e < \frac{1}{2}(H_1-H_2)$. Thus, if we change the value of $J_e$, the ground state does not change as long as the new value is on the same side of $\frac{1}{2}(H_1-H_2)$ as the old one. 

Let us say that two edges are ``neighbors'' of each other if they share one common endpoint. It is not hard to see that $|H_1-H_2|$ is at most the sum of $|J_f|$ over all edges $f$ that are neighbors of $e$. Let $S_e$ denote this sum. We will say that $e$ is a ``special edge'' if $J_e > S_e + 2$. Note that if $e=\{i,j\}$ is special, then $J_e >0$ and $\sigma_i =\sigma_j$ for the ground state $\sigma$.

%Let $D$ denote the maximum degree of $G$. Then 

It is easy to see that one can choose a subset $K\subseteq \partial A$ such that no two edges in $K$ are neighbors of each other or have a common neighbor, and $|K|\ge c|\partial A|$, where $c>0$ depends only on the maximum degree of $G$. 

We make two important observations about $K$. First, note that the events $\{J_e > S_e +2\}$, as $e$ ranges over $K$, are independent. Thus, if $X$ denotes the number of special edges in $K$, then $X$ is a sum of independent Bernoulli random variables. Moreover, it is not hard to see that $\ee(X) \ge a |K|$ for some constant $a>0$ depending only on the maximum degree of $G$.

Next, we claim that $\Delta(A) \ge 2X$. To see this, note that if we replace $J_e$ by $J_e -1$ for any special edge $e = \{i,j\}\in K$, then the ground state does not change. But all other special edges in $K$ remain special even after this operation, since no two edges in $K$ are neighbors of each other. Thus, we can repeat this substitution successively for each special edge in $K$, keeping the ground state unchanged. 

Let $\sigma$ denote the ground state in the environment $J$. Let $J'$ denote the new environment obtained above, and $\sigma'$ denote the state obtained by overturning all the spins in $A$. Then by the conclusion of the previous paragraph, $H_{J'}(\sigma)\le H_{J'}(\sigma')$. But note that $\sigma_i\sigma_j=1$ for every special edge $\{i,j\}$. Thus,
\begin{align*}
H_{J'}(\sigma') - H_{J'}(\sigma) &= 2\sum_{\{i,j\}\in \partial A} J'_{ij} \sigma_i \sigma_j\\
&= 2\sum_{\{i,j\}\in \partial A} J_{ij} \sigma_i \sigma_j - 2X\\
&= \Delta(A) - 2X.
\end{align*}
This proves that $\Delta(A)\ge 2X$. By the observations about $X$ made above, it is now easy to complete the proof (e.g., by Hoeffding's concentration inequality).

\subsection{Proof of Theorem \ref{dropletthm}}\label{dropletthmpf}
Consider the system perturbed by the second kind of perturbation, with parameter $p$. Let $X$ be the number of edges where $J_e$ is replaced by an independent copy $J_e'$. Then $X$ is a Binomial$(|E|, p)$ random variable. A different way to cause the same perturbation is to first generate $X$ from the Binomial$(|E|, p)$ distribution, and then pick $X$ distinct edges at random and replace the couplings by independent copies. Let $e_1,e_2,\ldots,e_X$ denote these edges. Let $\sigma^0 =\sigma$ be the original ground state, and $\sigma^k$ be the ground state after replacing $J_{e_1},\ldots,J_{e_k}$ by independent copies. Let $\sigma' = \sigma^X$ be the ground state after completing the whole replacement process. 

Let $R(p)$ denote the site overlap between $\sigma$ and $\sigma'$. Note that 
\begin{align*}
R(p)^2 &= \frac{1}{|V^\circ|^2} \sum_{i,j\in V^\circ} \sigma_i \sigma_j \sigma_i'\sigma_j'=  \frac{2}{|V^\circ|^2} \sum_{i,j\in V^\circ} \biggl(\frac{1}{2}- 1_{\{\sigma_i \sigma_j \ne \sigma_i'\sigma_j'\}}\biggr),
\end{align*}
which implies that
\begin{align}\label{rpeq}
\ee(R(p)^2) = \frac{2}{|V^\circ|^2} \sum_{i,j\in V^\circ} \biggl(\frac{1}{2}- \pp(\sigma_i \sigma_j \ne \sigma_i'\sigma_j')\biggr)
\end{align}
Now note that 
\begin{align}\label{psigma}
\pp(\sigma_i \sigma_j \ne \sigma_i'\sigma_j'|X) &\le \sum_{k=1}^X \pp(\sigma_i^{k-1}\sigma_j^{k-1} \ne \sigma_i^k \sigma_j^k|X). 
\end{align}
Let $\tilde{\sigma}$ be the ground state after replacing the weight on one uniformly chosen edge by an independent copy in the original system. Given $X$, $\sigma^{k-1}$ has the same law as $\sigma$ for any $1\le k\le X$. Given $X$ and $\sigma^{k-1}$, $e_k$ is uniformly distributed on $E$. Thus, given $X$, $(\sigma^{k-1}, \sigma^k)$ has the same distribution as $(\sigma, \tilde{\sigma})$. This shows that for any $1\le k\le X$, 
\[
\pp(\sigma_i^{k-1}\sigma_j^{k-1} \ne \sigma_i^k \sigma_j^k|X) = \pp(\sigma_i \sigma_j \ne \tilde{\sigma}_i \tilde{\sigma}_j). 
\]
Plugging this into \eqref{psigma}, we get
\begin{align*}
\pp(\sigma_i \sigma_j \ne \sigma_i'\sigma_j'|X) &\le X \pp(\sigma_i \sigma_j \ne \tilde{\sigma}_i \tilde{\sigma}_j). 
\end{align*}
Taking expectation on both sides gives
\begin{align*}
\pp(\sigma_i \sigma_j \ne \sigma_i'\sigma_j') &\le |E| p \pp(\sigma_i \sigma_j \ne \tilde{\sigma}_i \tilde{\sigma}_j).
\end{align*}
Combining this with \eqref{rpeq}, we get
\begin{align*}
\sum_{i,j\in V^\circ}  \pp(\sigma_i \sigma_j \ne \tilde{\sigma}_i \tilde{\sigma}_j) &\ge \frac{1}{|E|p} \sum_{i,j\in V^\circ} \pp(\sigma_i \sigma_j \ne \sigma_i'\sigma_j')\\
&=\frac{|V^\circ|^2}{2|E|p} (1-\ee(R(p)^2)).
\end{align*}
Applying Theorem \ref{mainthm} to the right side gives
\begin{align*}
\sum_{i,j\in V^\circ}  \pp(\sigma_i \sigma_j \ne \tilde{\sigma}_i \tilde{\sigma}_j) &\ge \frac{|V^\circ|^2}{2|E|p} \biggl(1- \frac{1}{|V^\circ|} - \frac{C (|V^\circ| p^{-\beta} + |B|^2 p^{-2\delta})}{|V^\circ|^2}\biggr). 
\end{align*}
Choosing 
\[
p = c \max\biggl\{|V^\circ|^{-\frac{1}{\beta}}, \biggl(\frac{|B|}{|V^\circ|}\biggr)^{\frac{1}{\delta}}\biggr\}
\]
 for some sufficiently large  $c$  (depending only on $\alpha$, $\beta$, $\gamma$ and $\delta$), and assuming that $|V^\circ|$ and $\frac{|V^\circ|}{|B|}$ are sufficiently large (again, depending only on $\alpha$, $\beta$, $\gamma$ and $\delta$), we can ensure that the term within the brackets on the right side above is at least $\frac{1}{2}$. (To see that $|V^\circ|$ and $\frac{|V^\circ|}{|B|}$ can be assumed to be sufficiently large, notice the following. For any $e$, $|D(e)|\ge 1$, since at least one spin is flipped. Thus, $\sum_{e\in E}\E|D(e)|\ge |E|$. Since $G$ is a connected graph, $|E|\ge |V|\ge |V|^\circ$. Thus, the claimed inequality holds automatically if at least one of $|V^\circ|$ and $\frac{|V^\circ|}{|B|}$ is smaller than some given constant.) Thus, if  $|V^\circ|$ and $\frac{|V^\circ|}{|B|}$ are large enough, then 
\begin{align*}
\sum_{i,j\in V^\circ}  \pp(\sigma_i \sigma_j \ne \tilde{\sigma}_i \tilde{\sigma}_j) &\ge \frac{C|V^\circ|^2}{|E| \max\bigl\{|V^\circ|^{-\frac{1}{\beta}}, \bigl(\frac{|B|}{|V^\circ|}\bigr)^{\frac{1}{\delta}}\bigr\}},
\end{align*}
for some $C>0$ that depends only on $\alpha$, $\beta$, $\gamma$ and $\delta$. But the number of pairs $(i,j)$ such that $\sigma_i \sigma_j \ne \tilde{\sigma}_i \tilde{\sigma}_j$ is equal to $(|V^\circ|-|A|)|A|$, where $A$ is the set of sites where $\sigma$ disagrees with $\tilde{\sigma}$ (taking the smaller of two sets if $B=\emptyset$). Thus,
\[
\ee|A| \ge \frac{1}{|V^\circ|} \ee[(|V^\circ|-|A|)|A|] = \frac{1}{|V^\circ|}\sum_{i,j\in V^\circ}  \pp(\sigma_i \sigma_j \ne \tilde{\sigma}_i \tilde{\sigma}_j).
\]
Combining with the previous display completes the proof of Theorem \ref{dropletthm}, because $A = D(e)$ for an edge $e$ chosen uniformly at random.

%Let $J'$ denote the configuration where $J_e' = -J_e$ for each $e\in E'$ and $J_e' = J_e$ for each $e\notin E'$. Since the law $\mu$ is symmetric and the $J_e$'s are independent, $J'$ has the same law as $J$. Now, it is easy to see that for any realization of $J$, $\phi_i^{j-}(J) = \phi_i^{j+}(J')$. 

\subsection{Proof of Theorem \ref{decaythm}}\label{decaythmpf}
Consider the EA model on the cube $V'$ of side-length $100L$ centered at the origin, with periodic boundary condition. Let $U$ be the cube of side-length $2L$ centered at $i$, so that $U\subseteq (V')^\circ$. Let $k := 4L e_1$, where $e_1 = (1,0,\ldots,0)$. Note that $k\in (V')^\circ\setminus U$. Then, it is easy to see that $\partial U$ can be written as the union of two disjoint sets $S_1$ and $S_2$, satisfying the conditions that 
\begin{itemize}
\item there is a path $P_1$ from $0$ to $k$ consisting of edges whose endpoints are all at a graph distance at least $2L$ from $S_1$, 
\item there is a path $P_2$ from $0$ to $k$ consisting of edges whose endpoints are all at a graph distance at least $2L$ from~$S_2$, and
\item the lengths of $P_1$ and $P_2$ are at most $16L$.
\end{itemize}
For example, one can take $S_1$ and $S_2$ to be the top and bottom halves of $\partial U$. See Figure~\ref{illus} for an illustration. 

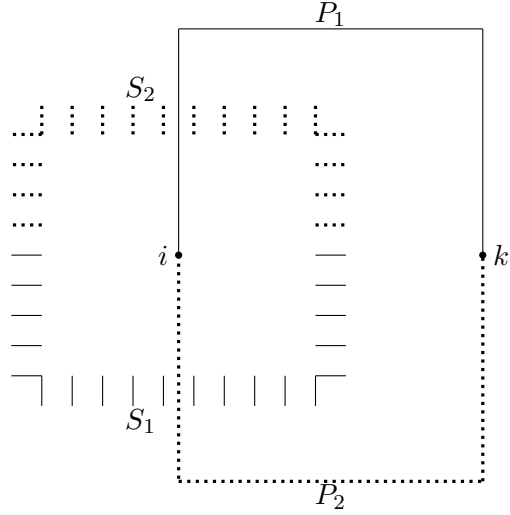
\begin{figure}[t]
\begin{center}
\begin{tikzpicture}[scale = 1]
\draw [fill] (0,0) circle [radius = 0.04] node [black,left] {$0$};
\draw [fill] (4,0) circle [radius = 0.04] node [black,right] {$k$};
\draw (1.8, 0) to (2.2,0);
\draw (1.8, -.4) to (2.2,-.4);
\draw (1.8, -.8) to (2.2,-.8);
\draw (1.8, -1.2) to (2.2,-1.2);
\draw (1.8, -1.6) to (2.2,-1.6);
\draw (-1.8, 0) to (-2.2,0);
\draw (-1.8, -.4) to (-2.2,-.4);
\draw (-1.8, -.8) to (-2.2,-.8);
\draw (-1.8, -1.2) to (-2.2,-1.2);
\draw (-1.8, -1.6) to (-2.2,-1.6);
\draw (-1.8, -1.6) to (-1.8,-2);
\draw (-1.4, -1.6) to (-1.4,-2);
\draw (-1, -1.6) to (-1,-2);
\draw (-.6, -1.6) to (-.6,-2);
\draw (-.2, -1.6) to (-.2,-2);
\draw (.2, -1.6) to (.2,-2);
\draw (.6, -1.6) to (.6,-2);
\draw (1, -1.6) to (1,-2);
\draw (1.4, -1.6) to (1.4,-2);
\draw (1.8, -1.6) to (1.8,-2);
\draw (0,0) to (0, 3);
\draw (0,3) to (4, 3);
\draw (4,3) to (4,0);
\draw (2,3.2) node {$P_1$};
\draw (-.5,-2.2) node {$S_1$};

\draw [very thick, dotted] (1.8, .4) to (2.2,.4);
\draw [very thick, dotted]  (1.8, .8) to (2.2,.8);
\draw [very thick, dotted]  (1.8, 1.2) to (2.2,1.2);
\draw [very thick, dotted]  (1.8, 1.6) to (2.2,1.6);
\draw [very thick, dotted]  (-1.8, .4) to (-2.2,.4);
\draw [very thick, dotted]  (-1.8, .8) to (-2.2,.8);
\draw [very thick, dotted]  (-1.8, 1.2) to (-2.2,1.2);
\draw [very thick, dotted]  (-1.8, 1.6) to (-2.2,1.6);
\draw [very thick, dotted]  (-1.8, 1.6) to (-1.8,2);
\draw [very thick, dotted]  (-1.4, 1.6) to (-1.4,2);
\draw [very thick, dotted]  (-1, 1.6) to (-1,2);
\draw [very thick, dotted]  (-.6, 1.6) to (-.6,2);
\draw [very thick, dotted]  (-.2, 1.6) to (-.2,2);
\draw [very thick, dotted]  (.2, 1.6) to (.2,2);
\draw [very thick, dotted]  (.6, 1.6) to (.6,2);
\draw [very thick, dotted]  (1, 1.6) to (1,2);
\draw [very thick, dotted]  (1.4, 1.6) to (1.4,2);
\draw [very thick, dotted]  (1.8, 1.6) to (1.8,2);

\draw [very thick, dotted] (0,0) to (0, -3);
\draw [very thick, dotted] (0,-3) to (4, -3);
\draw [very thick, dotted] (4,-3) to (4,0);
\draw (2,-3.2) node {$P_2$};
\draw (-.5,2.2) node {$S_2$};

\end{tikzpicture}
\caption{Schematic illustration of the sets $S_1, S_2$ and the paths $P_1, P_2$.\label{illus}}
\end{center}
\end{figure}

Let $J'$ be the environment obtained from $J$ by replacing $J_e$ by $-J_e$ for each $e\in S_1$, and $J''$ be the environment obtained from $J$ by replacing $J_e$  by $-J_e$ for each $e\in \partial U$. Let $\sigma$, $\sigma'$ and $\sigma''$ be the ground states in the environments $J$, $J'$ and $J''$. 

Let $e = \{u,v\}$ be an edge in the path $P_1$. By the properties of $P_1$ listed above, the  cube $W$ of side-length $L$ centered at $u$ does not intersect $S_1$. Let $Q$ be the event that $\tau_u\tau_v$ is {\it not} the same for all boundary conditions on $W$, where $\tau$ denotes the ground state for the EA model on $W$. Then the event $\{\sigma_u \sigma_v \ne \sigma'_u \sigma_v'\}$ implies $Q$, and hence
\[
\pp(Q) \ge \pp(\sigma_u \sigma_v \ne \sigma'_u \sigma_v'). 
\]
Now, we can define an event $Q$ as above for every edge in $P_1$, and they all have the same probability due to the periodic boundary condition. Thus,
\begin{align*}
|P_1| \pp(Q) &\ge \sum_{\{u,v\}\in P_1} \pp(\sigma_u \sigma_v \ne \sigma'_u \sigma_v')\\
&\ge \pp(\sigma_0\sigma_k\ne \sigma_0'\sigma_k'),
\end{align*}
where the last inequality follows from the observation that if $\sigma_0\sigma_k \ne \sigma_0'\sigma_k'$, then we must have that $\sigma_u\sigma_v \ne \sigma_u'\sigma_v'$ for some edge $\{u,v\}\in P_1$. But by assumption, $|P_1|\le 16L$. Thus,
\[
\pp(Q)\ge \frac{\pp(\sigma_0\sigma_k\ne \sigma_0'\sigma_k')}{16L}. 
\]
Similarly, working with $P_2$ and $S_2$ instead of $P_1$ and $S_1$, we get that 
\[
\pp(Q)\ge \frac{\pp(\sigma_0'\sigma_k'\ne \sigma_0''\sigma_k'')}{16L}. 
\]
Combining, we get
\begin{align*}
\pp(Q)&\ge \frac{\pp(\sigma_0\sigma_k\ne \sigma_0'\sigma_k')+ \pp(\sigma_0'\sigma_k'\ne \sigma_0''\sigma_k'')}{32L}\\
&\ge  \frac{\pp(\sigma_0\sigma_k\ne \sigma_0''\sigma_k'')}{32L}.
\end{align*}
But, the environment $J''$ satisfies $J_e'' = -J_e$ for all $e\in \partial U$ and $J_e'' = J_e$ for all $e\notin \partial U$. Moreover, $0\in U$ and $k\notin U$. Thus, with probability one, $\sigma_0'' \sigma_k'' = -\sigma_0\sigma_k$. Plugging this information into the above inequality, we get
\[
\pp(Q) \ge \frac{1}{32L}. 
\]
Since the event $Q$ corresponding to the edge $\{0,j\}$ is the event $\me$ from the theorem statement, this completes the proof.

\bibliographystyle{abbrvnat}

\bibliography{myrefs}

\begin{thebibliography}{64}
\providecommand{\natexlab}[1]{#1}
\providecommand{\url}[1]{\texttt{#1}}
\expandafter\ifx\csname urlstyle\endcsname\relax
  \providecommand{\doi}[1]{doi: #1}\else
  \providecommand{\doi}{doi: \begingroup \urlstyle{rm}\Url}\fi

\bibitem[Aizenman and Wehr(1989)]{aizenmanwehr89}
M.~Aizenman and J.~Wehr.
\newblock Rounding of first-order phase transitions in systems with quenched
  disorder.
\newblock \emph{Physical Review Letters}, 62\penalty0 (21):\penalty0 2503,
  1989.

\bibitem[Aizenman and Wehr(1990)]{aizenmanwehr90}
M.~Aizenman and J.~Wehr.
\newblock Rounding effects of quenched randomness on first-order phase
  transitions.
\newblock \emph{Communications in Mathematical Physics}, 130\penalty0
  (3):\penalty0 489--528, 1990.

\bibitem[Arguin and Damron(2014)]{arguindamron14}
L.-P. Arguin and M.~Damron.
\newblock On the number of ground states of the {E}dwards--{A}nderson spin
  glass model.
\newblock \emph{Annales de l'IHP Probabilit{\'e}s et Statistiques}, 50\penalty0
  (1):\penalty0 28--62, 2014.

\bibitem[Arguin and Hanson(2020)]{arguinhanson20}
L.-P. Arguin and J.~Hanson.
\newblock On absence of disorder chaos for spin glasses on $\zz^d$.
\newblock \emph{Electronic Communications in Probability}, 25:\penalty0 1--12,
  2020.

\bibitem[Arguin et~al.(2010)Arguin, Damron, Newman, and Stein]{arguinetal10}
L.-P. Arguin, M.~Damron, C.~M. Newman, and D.~L. Stein.
\newblock Uniqueness of ground states for short-range spin glasses in the
  half-plane.
\newblock \emph{Communications in Mathematical Physics}, 300\penalty0
  (3):\penalty0 641--657, 2010.

\bibitem[Arguin et~al.(2014)Arguin, Newman, Stein, and Wehr]{arguinetal14}
L.-P. Arguin, C.~M. Newman, D.~L. Stein, and J.~Wehr.
\newblock Fluctuation bounds for interface free energies in spin glasses.
\newblock \emph{Journal of Statistical Physics}, 156\penalty0 (2):\penalty0
  221--238, 2014.

\bibitem[Arguin et~al.(2016)Arguin, Newman, Stein, and Wehr]{arguinetal16}
L.-P. Arguin, C.~M. Newman, D.~L. Stein, and J.~Wehr.
\newblock Zero-temperature fluctuations in short-range spin glasses.
\newblock \emph{Journal of Statistical Physics}, 163\penalty0 (5):\penalty0
  1069--1078, 2016.

\bibitem[Arguin et~al.(2019)Arguin, Newman, and Stein]{arguinetal19}
L.-P. Arguin, C.~M. Newman, and D.~L. Stein.
\newblock A relation between disorder chaos and incongruent states in spin
  glasses on $\zz^d$.
\newblock \emph{Communications in Mathematical Physics}, 367\penalty0
  (3):\penalty0 1019--1043, 2019.

\bibitem[Arguin et~al.(2021)Arguin, Newman, and Stein]{arguinetal21}
L.-P. Arguin, C.~M. Newman, and D.~L. Stein.
\newblock Ground state stability in two spin glass models.
\newblock In \emph{In and Out of Equilibrium 3: Celebrating Vladas
  Sidoravicius}, pages 17--25. Springer, 2021.

\bibitem[Auffinger and Chen(2016)]{auffingerchen16}
A.~Auffinger and W.-K. Chen.
\newblock Universality of chaos and ultrametricity in mixed $p$-spin models.
\newblock \emph{Communications on Pure and Applied Mathematics}, 69\penalty0
  (11):\penalty0 2107--2130, 2016.

\bibitem[Ben~Arous et~al.(2020)Ben~Arous, Subag, and Zeitouni]{benarousetal20}
G.~Ben~Arous, E.~Subag, and O.~Zeitouni.
\newblock Geometry and temperature chaos in mixed spherical spin glasses at low
  temperature: the perturbative regime.
\newblock \emph{Communications on Pure and Applied Mathematics}, 73\penalty0
  (8):\penalty0 1732--1828, 2020.

\bibitem[Berger and Tessler(2017)]{bergertessler17}
N.~Berger and R.~J. Tessler.
\newblock No percolation in low temperature spin glass.
\newblock \emph{Electronic Journal of Probability}, 22:\penalty0 1--19, 2017.

\bibitem[Bollob{\'a}s and Leader(1991)]{bollobasleader91}
B.~Bollob{\'a}s and I.~Leader.
\newblock Edge-isoperimetric inequalities in the grid.
\newblock \emph{Combinatorica}, 11\penalty0 (4):\penalty0 299--314, 1991.

\bibitem[Bray and Moore(1985)]{braymoore85}
A.~J. Bray and M.~A. Moore.
\newblock Critical behavior of the three-dimensional {I}sing spin glass.
\newblock \emph{Physical Review B}, 31\penalty0 (1):\penalty0 631, 1985.

\bibitem[Bray and Moore(1987)]{braymoore87}
A.~J. Bray and M.~A. Moore.
\newblock Chaotic nature of the spin-glass phase.
\newblock \emph{Physical Review Letters}, 58\penalty0 (1):\penalty0 57, 1987.

\bibitem[Chatterjee(2009)]{chatterjee09}
S.~Chatterjee.
\newblock Disorder chaos and multiple valleys in spin glasses.
\newblock \emph{arXiv preprint arXiv:0907.3381}, 2009.

\bibitem[Chatterjee(2014)]{chatterjee14}
S.~Chatterjee.
\newblock \emph{Superconcentration and Related Topics}.
\newblock Springer, Cham, 2014.

\bibitem[Chen(2013)]{chen13}
W.-K. Chen.
\newblock Disorder chaos in the {S}herrington--{K}irkpatrick model with
  external field.
\newblock \emph{Annals of Probability}, 41\penalty0 (5):\penalty0 3345--3391,
  2013.

\bibitem[Chen(2014)]{chen14}
W.-K. Chen.
\newblock Chaos in the mixed even-spin models.
\newblock \emph{Communications in Mathematical Physics}, 328\penalty0
  (3):\penalty0 867--901, 2014.

\bibitem[Chen and Panchenko(2017)]{chenpanchenko17}
W.-K. Chen and D.~Panchenko.
\newblock Temperature chaos in some spherical mixed $p$-spin models.
\newblock \emph{Journal of Statistical Physics}, 166\penalty0 (5):\penalty0
  1151--1162, 2017.

\bibitem[Chen and Panchenko(2018)]{chenpanchenko18}
W.-K. Chen and D.~Panchenko.
\newblock Disorder chaos in some diluted spin glass models.
\newblock \emph{Annals of Applied Probability}, 28\penalty0 (3):\penalty0
  1356--1378, 2018.

\bibitem[Chen and Sen(2017)]{chensen17}
W.-K. Chen and A.~Sen.
\newblock Parisi formula, disorder chaos and fluctuation for the ground state
  energy in the spherical mixed $p$-spin models.
\newblock \emph{Communications in Mathematical Physics}, 350\penalty0
  (1):\penalty0 129--173, 2017.

\bibitem[Chen et~al.(2015)Chen, Hsieh, Hwang, and Sheu]{chenetal15}
W.-K. Chen, H.-W. Hsieh, C.-R. Hwang, and Y.-C. Sheu.
\newblock Disorder chaos in the spherical mean-field model.
\newblock \emph{Journal of Statistical Physics}, 160\penalty0 (2):\penalty0
  417--429, 2015.

\bibitem[Chen et~al.(2017)Chen, Dey, and Panchenko]{chenetal17}
W.-K. Chen, P.~Dey, and D.~Panchenko.
\newblock Fluctuations of the free energy in the mixed $p$-spin models with
  external field.
\newblock \emph{Probability Theory and Related Fields}, 168\penalty0
  (1):\penalty0 41--53, 2017.

\bibitem[Chen et~al.(2018)Chen, Handschy, and Lerman]{chenetal18}
W.-K. Chen, M.~Handschy, and G.~Lerman.
\newblock On the energy landscape of the mixed even $p$-spin model.
\newblock \emph{Probability Theory and Related Fields}, 171\penalty0
  (1):\penalty0 53--95, 2018.

\bibitem[Contucci and Giardin{\`a}(2013)]{contuccigiardina13}
P.~Contucci and C.~Giardin{\`a}.
\newblock \emph{Perspectives on spin glasses}.
\newblock Cambridge University Press, 2013.

\bibitem[Contucci et~al.(2006)Contucci, Giardin\`a, Giberti, and
  Vernia]{contuccietal06}
P.~Contucci, C.~Giardin\`a, C.~Giberti, and C.~Vernia.
\newblock Overlap equivalence in the {E}dwards--{A}nderson model.
\newblock \emph{Physical Review Letters}, 96\penalty0 (21):\penalty0 217204,
  2006.

\bibitem[Contucci et~al.(2007)Contucci, Giardin\`a, Giberti, Parisi, and
  Vernia]{contuccietal07}
P.~Contucci, C.~Giardin\`a, C.~Giberti, G.~Parisi, and C.~Vernia.
\newblock Ultrametricity in the {E}dwards--{A}nderson model.
\newblock \emph{Physical Review Letters}, 99\penalty0 (5):\penalty0 057206,
  2007.

\bibitem[Cotar et~al.(2018)Cotar, Jahnel, and K{\"u}lske]{cotaretal18}
C.~Cotar, B.~Jahnel, and C.~K{\"u}lske.
\newblock Extremal decomposition for random {G}ibbs measures: from general
  metastates to metastates on extremal random {G}ibbs measures.
\newblock \emph{Electronic Communications in Probability}, 23:\penalty0 1--12,
  2018.

\bibitem[Dario et~al.(2021)Dario, Harel, and Peled]{darioetal21}
P.~Dario, M.~Harel, and R.~Peled.
\newblock Quantitative disorder effects in low-dimensional spin systems.
\newblock \emph{arXiv preprint arXiv:2101.01711}, 2021.

\bibitem[Ding and Xia(2021)]{dingxia21}
J.~Ding and J.~Xia.
\newblock Exponential decay of correlations in the two-dimensional random field
  {I}sing model.
\newblock \emph{Inventiones Mathematicae}, 224\penalty0 (3):\penalty0
  999--1045, 2021.

\bibitem[Ding and Zhuang(2021)]{dingzhuang21}
J.~Ding and Z.~Zhuang.
\newblock Long range order for random field {I}sing and {P}otts models.
\newblock \emph{arXiv preprint arXiv:2110.04531}, 2021.

\bibitem[Edwards and Anderson(1975)]{edwardsanderson75}
S.~F. Edwards and P.~W. Anderson.
\newblock Theory of spin glasses.
\newblock \emph{Journal of Physics F: Metal Physics}, 5\penalty0 (5):\penalty0
  965, 1975.

\bibitem[Eldan(2020)]{eldan20}
R.~Eldan.
\newblock A simple approach to chaos for $p$-spin models.
\newblock \emph{Journal of Statistical Physics}, 181\penalty0 (4):\penalty0
  1266--1276, 2020.

\bibitem[Fisher and Huse(1986)]{fisherhuse86}
D.~S. Fisher and D.~A. Huse.
\newblock Ordered phase of short-range {I}sing spin-glasses.
\newblock \emph{Physical Review Letters}, 56\penalty0 (15):\penalty0 1601,
  1986.

\bibitem[Fisher and Huse(1988)]{fisherhuse88}
D.~S. Fisher and D.~A. Huse.
\newblock Equilibrium behavior of the spin-glass ordered phase.
\newblock \emph{Physical Review B}, 38\penalty0 (1):\penalty0 386, 1988.

\bibitem[Gamarnik(2021)]{gamarnik21}
D.~Gamarnik.
\newblock The overlap gap property: {A} topological barrier to optimizing over
  random structures.
\newblock \emph{Proceedings of the National Academy of Sciences}, 118\penalty0
  (41):\penalty0 e2108492118, 2021.

\bibitem[Garban and Steif(2014)]{garbansteif14}
C.~Garban and J.~E. Steif.
\newblock \emph{Noise sensitivity of {B}oolean functions and percolation}.
\newblock Cambridge University Press, 2014.

\bibitem[Garban et~al.(2010)Garban, Pete, and Schramm]{garbanetal10}
C.~Garban, G.~Pete, and O.~Schramm.
\newblock The {F}ourier spectrum of critical percolation.
\newblock \emph{Acta Mathematica}, 205\penalty0 (1):\penalty0 19--104, 2010.

\bibitem[Garey and Johnson(1979)]{gareyjohnson79}
M.~R. Garey and D.~S. Johnson.
\newblock \emph{Computers and Intractability: A guide to the theory of
  NP-completeness}.
\newblock Freeman, San Francisco, CA, 1979.

\bibitem[Huang and Sellke(2021)]{huangsellke21}
B.~Huang and M.~Sellke.
\newblock Tight {L}ipschitz hardness for optimizing mean field spin glasses.
\newblock \emph{arXiv preprint arXiv:2110.07847}, 2021.

\bibitem[Huse and Fisher(1987)]{husefisher87}
D.~A. Huse and D.~S. Fisher.
\newblock Pure states in spin glasses.
\newblock \emph{Journal of Physics A: Mathematical and General}, 20\penalty0
  (15):\penalty0 L997, 1987.

\bibitem[Imbrie(1985)]{imbrie85}
J.~Z. Imbrie.
\newblock The ground state of the three-dimensional random-field {I}sing model.
\newblock \emph{Communications in Mathematical Physics}, 98\penalty0
  (2):\penalty0 145--176, 1985.

\bibitem[Imry and Ma(1975)]{imryma75}
Y.~Imry and S.-k. Ma.
\newblock Random-field instability of the ordered state of continuous symmetry.
\newblock \emph{Physical Review Letters}, 35\penalty0 (21):\penalty0 1399,
  1975.

\bibitem[Krzakala and Martin(2000)]{krzakalamartin00}
F.~Krzakala and O.~C. Martin.
\newblock Spin and link overlaps in three-dimensional spin glasses.
\newblock \emph{Physical Review Letters}, 85\penalty0 (14):\penalty0 3013,
  2000.

\bibitem[Marinari and Parisi(2001)]{marinariparisi01}
E.~Marinari and G.~Parisi.
\newblock Effects of a bulk perturbation on the ground state of {3D} {I}sing
  spin glasses.
\newblock \emph{Physical Review Letters}, 86\penalty0 (17):\penalty0 3887,
  2001.

\bibitem[McMillan(1984)]{mcmillan84}
W.~L. McMillan.
\newblock Scaling theory of {I}sing spin glasses.
\newblock \emph{Journal of Physics C: Solid State Physics}, 17\penalty0
  (18):\penalty0 3179, 1984.

\bibitem[M{\'e}zard et~al.(1987)M{\'e}zard, Parisi, and Virasoro]{mezardetal87}
M.~M{\'e}zard, G.~Parisi, and M.~A. Virasoro.
\newblock \emph{{Spin glass theory and beyond: An Introduction to the Replica
  Method and Its Applications}}, volume~9.
\newblock World Scientific Publishing Company, 1987.

\bibitem[Newman and Stein(1992)]{newmanstein92}
C.~M. Newman and D.~L. Stein.
\newblock Multiple states and thermodynamic limits in short-ranged {I}sing
  spin-glass models.
\newblock \emph{Physical Review B}, 46\penalty0 (2):\penalty0 973, 1992.

\bibitem[Newman and Stein(1998)]{newmanstein98}
C.~M. Newman and D.~L. Stein.
\newblock Thermodynamic chaos and the structure of short-range spin glasses.
\newblock In \emph{Mathematical aspects of spin glasses and neural networks},
  pages 243--287. Springer, 1998.

\bibitem[Newman and Stein(2001)]{newmanstein01}
C.~M. Newman and D.~L. Stein.
\newblock Are there incongruent ground states in 2d {E}dwards--{A}nderson spin
  glasses?
\newblock \emph{Communications in Mathematical Physics}, 224\penalty0
  (1):\penalty0 205--218, 2001.

\bibitem[Newman and Stein(2003)]{newmanstein03}
C.~M. Newman and D.~L. Stein.
\newblock Ordering and broken symmetry in short-ranged spin glasses.
\newblock \emph{Journal of Physics: Condensed Matter}, 15\penalty0
  (32):\penalty0 R1319--R1364, 2003.

\bibitem[Palassini and Young(2000)]{palassiniyoung00}
M.~Palassini and A.~P. Young.
\newblock Nature of the spin glass state.
\newblock \emph{Physical Review Letters}, 85\penalty0 (14):\penalty0 3017,
  2000.

\bibitem[Panchenko(2013)]{panchenko13}
D.~Panchenko.
\newblock \emph{The {S}herrington--{K}irkpatrick model}.
\newblock Springer Science \& Business Media, 2013.

\bibitem[Panchenko(2016)]{panchenko16}
D.~Panchenko.
\newblock Chaos in temperature in generic $2p$-spin models.
\newblock \emph{Communications in Mathematical Physics}, 346\penalty0
  (2):\penalty0 703--739, 2016.

\bibitem[Parisi(2007)]{parisi07}
G.~Parisi.
\newblock Mean field theory of spin glasses: Statics and dynamics.
\newblock In \emph{Complex Systems}, volume~85 of \emph{Les Houches}, pages
  131--178. Elsevier, 2007.

\bibitem[Parisi and Ricci-Tersenghi(2000)]{parisiricci-tersenghi00}
G.~Parisi and F.~Ricci-Tersenghi.
\newblock On the origin of ultrametricity.
\newblock \emph{Journal of Physics A: Mathematical and General}, 33\penalty0
  (1):\penalty0 113, 2000.

\bibitem[Read(2014)]{read14}
N.~Read.
\newblock Short-range {I}sing spin glasses: {T}he metastate interpretation of
  replica symmetry breaking.
\newblock \emph{Physical Review E}, 90\penalty0 (3):\penalty0 032142, 2014.

\bibitem[Sherrington and Kirkpatrick(1975)]{sherringtonkirkpatrick75}
D.~Sherrington and S.~Kirkpatrick.
\newblock Solvable model of a spin-glass.
\newblock \emph{Physical Review Letters}, 35\penalty0 (26):\penalty0 1792,
  1975.

\bibitem[Shih et~al.(1990)Shih, Wu, and Kuo]{shihetal90}
W.-K. Shih, S.~Wu, and Y.-S. Kuo.
\newblock Unifying maximum cut and minimum cut of a planar graph.
\newblock \emph{IEEE Transactions on Computers}, 39\penalty0 (5):\penalty0
  694--697, 1990.

\bibitem[Stein(2016)]{stein16}
D.~L. Stein.
\newblock Frustration and fluctuations in systems with quenched disorder.
\newblock In \emph{Pwa90: A Lifetime of Emergence}, pages 169--186. World
  Scientific, 2016.

\bibitem[Subag(2017)]{subag17}
E.~Subag.
\newblock The geometry of the {G}ibbs measure of pure spherical spin glasses.
\newblock \emph{Inventiones Mathematicae}, 210\penalty0 (1):\penalty0 135--209,
  2017.

\bibitem[Talagrand(2010)]{talagrand10}
M.~Talagrand.
\newblock \emph{Mean field models for spin glasses: Volume I: Basic examples}.
\newblock Springer Science \& Business Media, 2010.

\bibitem[Talagrand(2011)]{talagrand11}
M.~Talagrand.
\newblock \emph{Mean Field Models for Spin Glasses: Volume II: Advanced
  Replica-Symmetry and Low Temperature}.
\newblock Springer Science \& Business Media, 2011.

\end{thebibliography}

\end{document}